\documentclass[11pt,letterpaper]{article}
\usepackage[margin=1in]{geometry}
\usepackage[T1]{fontenc}
\usepackage{lmodern}
\usepackage{amsmath,amssymb,amsthm,mathtools}
\usepackage[shortlabels]{enumitem}
\usepackage{microtype}
\usepackage{xcolor}
\usepackage{tikz}
\usetikzlibrary{arrows.meta,decorations.pathreplacing}
\usepackage[colorlinks=true,linkcolor=blue!50!black,citecolor=blue!50!black,urlcolor=blue!50!black]{hyperref}
\usepackage[nameinlink,capitalise]{cleveref}

\hypersetup{
  pdftitle={Tight Bounds for Randomized Generalized k-Server on Uniform Metrics},
  pdfauthor={Junho Hwang}
}

\newtheorem{theorem}{Theorem}[section]
\newtheorem{lemma}[theorem]{Lemma}
\newtheorem{corollary}[theorem]{Corollary}
\theoremstyle{definition}
\newtheorem{definition}[theorem]{Definition}
\theoremstyle{remark}
\newtheorem{remark}[theorem]{Remark}

\newcommand{\OPT}{\mathrm{OPT}}
\newcommand{\ALG}{\mathrm{ALG}}
\newcommand{\E}{\mathbb{E}}
\newcommand{\cX}{\mathcal{X}}
\newcommand{\cL}{\mathcal{L}}
\newcommand{\cH}{\mathcal{H}}
\newcommand{\cT}{\mathcal{T}}
\newcommand{\cU}{\mathcal{U}}
\newcommand{\cA}{\mathcal{A}}
\newcommand{\cQ}{\mathcal{Q}}
\newcommand{\Perm}{\mathfrak{S}}
\newcommand{\one}{\mathbf{1}}

\title{Tight Bounds for Randomized Generalized $k$-Server\\ on Uniform Metrics}
\author{Junho Hwang}
\date{}

\begin{document}

\maketitle
\thispagestyle{empty}

\begin{abstract}
\normalsize
In the generalized $k$-server problem, each of $k$ servers moves in its own metric space, and a request $(r_1,\dots,r_k)$ is served once some server $i$ stands at the point $r_i$ of its space.
We determine the randomized competitive ratio against an oblivious adversary when every space is a uniform metric.
With unit distances the ratio is $\Theta(k\log k)$: we give a $(k+1)\ln(2k)$-competitive algorithm and show that for large $k$ every algorithm has ratio at least $\frac1{20}k\ln k$, improving the bounds $O(k^2\log k)$ and $\Omega(k)$ of Bienkowski, Je\.z, and Schmidt.
When server $i$ pays $w_i$ per move, the ratio is $\exp(\Theta(k))$: we give a $20^k$-competitive algorithm, which also applies to weighted $k$-server on uniform metrics.
This improves the $\exp(O(k^2))$ bound of Bijoy, Mondal, and Chiplunkar and matches the $\Omega(2^k)$ lower bound of Ayyadevara and Chiplunkar up to the base of the exponent.
Both algorithms follow normalized weights exactly, through couplings whose expected movement is controlled by the decrease of the logarithm of the total weight.
The lower bound hides a random permutation among the solutions of random permutation constraints, which split into small, well-separated components, and then forbids one position to one server at a time.
\end{abstract}

\section{Introduction}\label{sec:intro}

In the $k$-server problem of Manasse, McGeoch, and Sleator~\cite{MMS90}, an online algorithm serves a sequence of requests, each a point of a metric space, by moving one of $k$ servers to the requested point, and it aims to minimize the total distance moved.
Koutsoupias and Taylor~\cite{KT04} introduced a far-reaching generalization, the \emph{generalized $k$-server problem}.
Each server $i$ moves in its own metric space $M_i$.
A request is a tuple $r=(r_1,\dots,r_k)$ with $r_i\in M_i$, and it is served once some server $i$ is at $r_i$.
The $k$-server problem is the special case in which all spaces are equal and every request has the form $(r,\dots,r)$; in the weighted $k$-server problem~\cite{New91}, the spaces are scaled copies of one metric space, again with requests of the form $(r,\dots,r)$.
Koutsoupias and Taylor showed that natural $k$-server algorithms fail already for two servers on two lines, the CNN problem, which they described as a problem that ``may act as a stepping stone towards building a robust (and less ad hoc) theory of online computation''.
For general metric spaces, competitive algorithms are known only for $k=2$~\cite{SS06,Sit14}.

\paragraph{Uniform metrics.}
Bansal, Eli\'a\v{s}, Koumoutsos, and Nederlof~\cite{BEKN18} obtained the first $f(k)$-competitive algorithms for all $k$, for the case where every $M_i$ is a uniform metric, in which distinct points are at distance~$1$.
Moving the servers from a configuration $x=(x_1,\dots,x_k)$ to a configuration $y$ then costs the number of coordinates in which $x$ and $y$ differ, and a request $r$ forbids exactly the configurations that differ from $r$ in every coordinate.
For deterministic algorithms they gave an $O(k2^k)$ upper bound by the polynomial method, which nearly matches the lower bound $2^k-1$ of Koutsoupias and Taylor~\cite{KT04}.
For randomized algorithms against an oblivious adversary, they gave an $O(k^3\log k)$-competitive algorithm and a lower bound of $\Omega(k/\log^2k)$.
Bienkowski, Je\.z, and Schmidt~\cite{BJS19} recast the phase structure of~\cite{BEKN18} as an online game on trees, the Hydra game, and improved the upper bound to $O(k^2\log k)$; they also improved the lower bound to $\Omega(k)$.
They note that their strategy for the Hydra game is optimal up to a constant factor on a class of trees that includes those arising in their reduction, so that improving $O(k^2\log k)$ requires a different approach~\cite[Section~5]{BJS19}.
Against adaptive online adversaries, the best competitive ratio of memoryless algorithms is $\Theta(k!)$~\cite{CFK20}.

\paragraph{Weighted uniform metrics.}
In the \emph{weighted} version, distinct points of $M_i$ are at distance $w_i$, so that a move of server $i$ costs $w_i$.
Its special case where the spaces are copies of one uniform metric and all requests have the form $(r,\dots,r)$ is weighted $k$-server on uniform metrics, introduced by Fiat and Ricklin~\cite{FR94}.
Here the deterministic competitive ratio is doubly exponential in $k$: Fiat and Ricklin gave a doubly exponential upper bound, Bansal, Eli\'a\v{s}, and Koumoutsos~\cite{BEK17} a doubly exponential lower bound, and Bansal et al.~\cite{BEKN18} extended the upper bound $2^{2^{k+3}}$ to generalized requests.
For randomized algorithms, the best ratio of memoryless algorithms against adaptive online adversaries is doubly exponential~\cite{CV20}, and Ayyadevara and Chiplunkar~\cite{AC21} proved a lower bound of $\Omega(2^k)$ for weighted $k$-server on uniform metrics that holds for all randomized algorithms against oblivious adversaries.
Ayyadevara, Chiplunkar, and Sharma~\cite{ACS24} reduced the design of a $2^{\mathrm{poly}(k)}$-competitive algorithm to one of two relaxed problems.
Bijoy, Mondal, and Chiplunkar~\cite{BMC26} gave the first algorithm that is not doubly exponential, with ratio $\exp(O(k^2))$, together with a sketch of the same bound for generalized requests.
They wrote that this almost solves a question of Bansal et al., whether generalized $k$-server on weighted uniform metrics admits an $\exp(k)$-competitive algorithm, and that ``the only gap left to close is between the $\Omega(2^k)$ and $\exp(O(k^2))$ bounds on the randomized competitive ratio''~\cite[Section~7]{BMC26}.

\subsection{Our results}

We determine the randomized competitive ratio for uniform metrics, with and without weights.
All logarithms are natural.

\begin{theorem}\label{thm:upper}
Let $M_1,\dots,M_k$ be uniform metric spaces with unit distances, finite or infinite, and let $a$ be an initial configuration.
There is a randomized online algorithm such that, for every finite request sequence $\rho$,
\[
\E[\ALG(\rho)]\le (k+1)\ln(2k)\cdot\OPT(\rho),
\]
where $\OPT(\rho)$ is the cost of an optimal offline solution that starts from $a$.
\end{theorem}

\begin{theorem}\label{thm:lower}
There is an absolute constant $k_0$ such that, for every $k\ge k_0$, every randomized online algorithm for generalized $k$-server on $k$ copies of the uniform metric on $k$ points has competitive ratio at least $\frac1{20}k\ln k$ against an oblivious adversary.
\end{theorem}

\begin{theorem}\label{thm:weighted}
Let $M_1,\dots,M_k$ be uniform metric spaces in which distinct points of $M_i$ are at distance $w_i$, where $0\le w_1\le\dots\le w_k$, and let $a$ be an initial configuration.
There is a randomized online algorithm such that, for every finite request sequence $\rho$,
\[
\E[\ALG(\rho)]\le 20^k\cdot\OPT(\rho)+10^kw_k .
\]
\end{theorem}

\begin{corollary}\label[corollary]{cor:main}
The randomized competitive ratio of generalized $k$-server is $\Theta(k\log k)$ on uniform metrics and $\exp(\Theta(k))$ on weighted uniform metrics.
The randomized competitive ratio of weighted $k$-server on uniform metrics is $\exp(\Theta(k))$.
\end{corollary}

\begin{proof}
The upper bounds are \cref{thm:upper,thm:weighted}, and \cref{thm:lower} is the lower bound with unit weights.
Weighted $k$-server on a uniform metric is the special case of generalized $k$-server on weighted uniform metrics described above, so the lower bound $\Omega(2^k)$ of Ayyadevara and Chiplunkar~\cite[Theorem~2]{AC21} applies to both weighted problems.
\end{proof}

\Cref{thm:upper} improves the upper bound by a factor of $k$, has no additive term, and holds for spaces of any cardinality.
\Cref{thm:lower} improves the lower bound by a factor of $\log k$; its construction uses spaces with $k$ points, while the construction of~\cite{BJS19} uses two points per space.
With \cref{thm:upper}, it shows that for large $k$ the ratio lies between $\frac1{20}k\ln k$ and $(1+o(1))k\ln k$.
\Cref{thm:weighted} holds for any number of points and any weights, and its constants do not depend on the ratios between the weights.

The algorithm of \cref{thm:upper} also serves requests that specify a set of points for each server: if the sets of every request have at most $B$ points in total, it is $(k+1)\ln(2B)$-competitive (\cref{thm:sets}).
This gives the ratio $(k+1)\ln(2k)$ for paging with heterogeneous cache slots~\cite{CHLPRSY23}, in which each request may be served only from certain slots, and $(k+1)\ln(2k\ell)$ for metrical service systems with $k$ servers on a uniform metric whose requests are sets of $\ell$ points~\cite{CV15}.
\Cref{tab:bounds} compares all bounds.

\begin{table}[t]
\centering
\small
\begin{tabular}{@{}lll@{}}
\hline
Problem & Previous bounds & This paper\\
\hline
generalized $k$-server, unit distances & $O(k^2\log k)$, $\Omega(k)$~\cite{BJS19} & $(k+1)\ln(2k)$, $\Omega(k\log k)$\\
generalized $k$-server with weights & $\exp(O(k^2))$~\cite{BMC26}, $\Omega(2^k)$~\cite{AC21} & $20^k$\\
weighted $k$-server & $\exp(O(k^2))$~\cite{BMC26}, $\Omega(2^k)$~\cite{AC21} & $20^k$\\
\hline
paging with heterogeneous cache slots & $O(k^2\log k)$, $\Omega(k)$~\cite{CHLPRSY23} & $(k+1)\ln(2k)$\\
$k$ servers, requests of $\ell\ge2$ points & $O(k^3\log\ell)$, $\Omega(\log(k\ell))$~\cite{CV15} & $(k+1)\ln(2k\ell)$\\
\hline
\end{tabular}
\caption{Competitive ratios of randomized algorithms on uniform metrics against oblivious adversaries, before and after this paper.
The last two rows are applications of \cref{thm:sets}.}
\label{tab:bounds}
\end{table}

\subsection{Techniques}\label{sec:techniques}

\paragraph{The algorithm for unit weights.}
The algorithms of~\cite{BEKN18,BJS19} work in phases.
Within a phase they stay in configurations that serve every request of the phase so far; a request that no such configuration serves is served separately and ends the phase, so the optimum moves in every completed phase.
Our algorithm gives up this invariant.
It keeps a weight on every configuration, initially concentrated on the starting configuration.
When a request $r$ arrives, a configuration that serves $r$ keeps its weight.
A configuration $x$ that fails $r$ passes a $1/(2k)$ fraction of its weight to each of the $k$ configurations obtained from $x$ by moving one server to its requested point, which we call the \emph{repairs} of $x$, and loses the remaining half (\cref{fig:update}).
The weights thereby account for lazy offline solutions, which move one server only when the request forces them to.
A lazy solution that has moved $m$ times keeps weight at least $(2k)^{-m}$, and some lazy solution moves at most $\OPT$ times.
The servers follow the normalized weights exactly.
If the current configuration fails the request, they move to a uniformly random one of its $k$ repairs with probability $1/2$, and otherwise to a fresh sample from the new distribution.
This keeps the configuration distributed exactly according to the weights, and its expected movement is at most $k+1$ times the decrease of the logarithm of the total weight, which is at most $\OPT\cdot\ln(2k)$ in total.
We build on the analysis of the weighted majority algorithm~\cite{LW94} and its extension to comparators that switch~\cite{HW98}, which Blum and Burch~\cite{BB00} used for metrical task systems, and on the shrinking dartboard algorithm of Geulen, V\"ocking, and Winkler~\cite{GVW10}, which follows a weight distribution exactly and resamples only to account for lost weight.

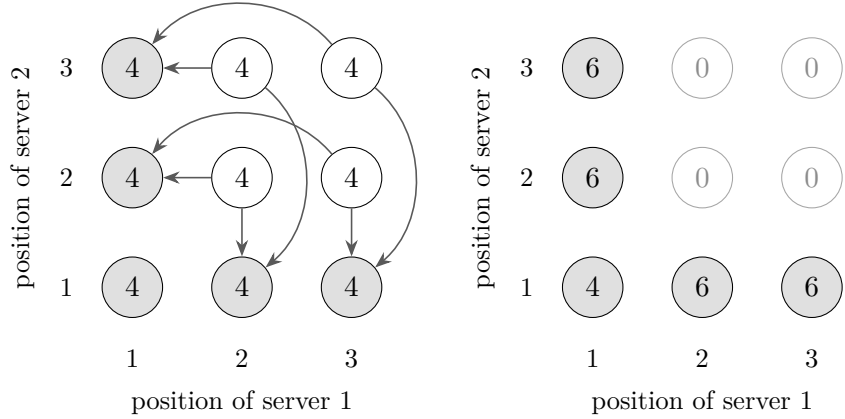
\begin{figure}[t]
\centering
\begin{tikzpicture}[x=1.45cm,y=1.45cm,
  serve/.style={circle,draw,fill=black!12,minimum size=8mm,inner sep=0pt},
  fail/.style={circle,draw,minimum size=8mm,inner sep=0pt},
  pass/.style={-{Stealth[length=2.2mm]},semithick,black!65}]
  \foreach \dx in {0,4.2} {
    \foreach \i in {1,2,3} {
      \node[font=\small] at (\dx+\i,0.35) {$\i$};
      \node[font=\small] at (\dx+0.4,\i) {$\i$};
    }
    \node[font=\small] at (\dx+2,-0.05) {position of server 1};
    \node[font=\small,rotate=90] at (\dx+0.0,2) {position of server 2};
  }
  \foreach \i/\j in {1/1,1/2,1/3,2/1,3/1} { \node[serve] (L\i\j) at (\i,\j) {$4$}; }
  \foreach \i/\j in {2/2,2/3,3/2,3/3} { \node[fail] (L\i\j) at (\i,\j) {$4$}; }
  \draw[pass] (L22) -- (L12);
  \draw[pass] (L23) -- (L13);
  \draw[pass] (L32) to[bend right=50] (L12);
  \draw[pass] (L33) to[bend right=50] (L13);
  \draw[pass] (L22) -- (L21);
  \draw[pass] (L32) -- (L31);
  \draw[pass] (L23) to[bend left=50] (L21);
  \draw[pass] (L33) to[bend left=50] (L31);
  \node[serve] at (4.2+1,1) {$4$};
  \foreach \i/\j in {1/2,1/3,2/1,3/1} { \node[serve] at (4.2+\i,\j) {$6$}; }
  \foreach \i/\j in {2/2,2/3,3/2,3/3} { \node[fail,draw=black!35,text=black!45] at (4.2+\i,\j) {$0$}; }
\end{tikzpicture}
\caption{One weight update for $k=2$ with three points per space, at the request $r=(1,1)$.
The configuration $(x_1,x_2)$ sits in column $x_1$ and row $x_2$, and weights are multiplied by $36$.
Left: an illustrative weight function, with weight $4$ on every configuration.
The shaded configurations serve $r$ and keep their weight.
Each white configuration fails $r$ and passes a quarter of its weight along each of its two arrows, to its repairs: server~1 moves to point~$1$ (leftward arrows) or server~2 moves to point~$1$ (downward arrows); the other half of its weight is discarded.
Right: after the request.
The total weight drops from $36$ to $28$, by half the weight of the white configurations, and the servers follow the new weights divided by $28$.}
\label{fig:update}
\end{figure}

\paragraph{The lower bound for unit weights.}
The lower bound of~\cite{BJS19} uses two points per space.
A configuration is then a binary string, and the request that equals the complement of a string $c$ forbids $c$ and nothing else.
The adversary removes the strings of a code of size $2^{\Omega(k)}$ and minimum distance $\Omega(k)$ one at a time, each time the one that the algorithm is most likely to occupy, and between removals it forbids configurations outside the remaining set until the algorithm is in that set with high probability.
Both steps use only the distribution of the algorithm's configuration, which the adversary computes from the algorithm and the requests so far, so the requests are independent of the algorithm's random choices, as for an oblivious adversary; the same holds for our adversary.
When $n$ strings remain, a removal thus moves the algorithm a distance $\Omega(k)$ with probability about $1/n$ or more, so over a phase the algorithm pays $\Omega(k)$ times the logarithm of the size of the code, while the optimum pays at most $k$.
To force $\omega(k)$, we therefore hide the target among $e^{\Omega(k\log k)}$ candidates, which needs more than a constant number of points per space; but then a single request forbids every configuration that disagrees with it in all coordinates, and candidates can no longer be removed one at a time.

Our adversary hides a uniformly random permutation $\sigma$ of $[k]$, viewed as a configuration on $k$ copies of the uniform metric on $[k]$, among a fixed set of candidate permutations.
A configuration serves the $k$ requests $(a,\dots,a)$ with $a\in[k]$ if and only if it is a permutation, and a permutation $x$ serves the request whose $i$-th coordinate is some $b\ne a$ and whose other coordinates are all $a$ if and only if $x_i\ne a$.
So, among permutations, this request forbids one position to one server: it removes the candidates that place server~$i$ at point~$a$.
The candidates are the solutions of random permutation constraints: for $h\approx0.9\ln k$ families $\pi^j=(\pi^j_1,\dots,\pi^j_k)$ of bijections of $[k]$, the adversary issues the $k$ requests $((\pi^j_1)^{-1}(a),\dots,(\pi^j_k)^{-1}(a))$ with $a\in[k]$, and a configuration $x$ serves all of them if and only if $(\pi^j_1(x_1),\dots,\pi^j_k(x_k))$ is a permutation.
The permutations that serve these \emph{setup requests} form the solution set of a random constraint satisfaction problem, and we call a configuration \emph{feasible} if it serves all requests so far.
We show that for some families, which the adversary fixes, there are $e^{\Omega(k\log k)}$ solutions, and they split into components at distance more than $0.77k$ from each other, each of which has fewer than $2^k$ configurations, in which more than $0.92k$ servers each take a single position.
In each phase, the adversary draws $\sigma$ uniformly from the solutions and then issues $\Theta(k^2\log k)$ requests, each forbidding a uniformly random position that $\sigma$ avoids.
Given the requests, $\sigma$ is uniform on the feasible configurations, and the expected inverse of their number stays far below $2^{-k}$, so at each request the algorithm is in a component without $\sigma$ with probability close to~$1$.
In more than $0.69k$ coordinates, all configurations of that component agree and differ from $\sigma$, so the request eliminates the whole component with probability more than $0.69/k$; repeating earlier requests then drives the algorithm back into the feasible set, a distance of more than $0.77k$ away (\cref{fig:clusters}).
Over a phase the optimum pays at most $k$, while the algorithm pays at least $\frac1{20}k^2\ln k$ in expectation.

The core of the proof is the component structure, which resembles the clusters of solutions of random constraint satisfaction problems~\cite{AC08}.
It reduces to bounding, for a solution $x$ of the setup requests and a set $I$ of coordinates, the values that the coordinates in $I$ take in solutions that agree with $x$ outside $I$.
We bound the probabilities of rare events by first moments under the uniform law of the constraints, and the second moment of the number of solutions~\cite{AM06} shows that some families have many solutions and none of these events.
We build on the routine Confine of~\cite{BJS19}, which forbids configurations outside a target set until the algorithm is in that set with high probability, at no cost to an offline solution that stays in the set; we implement it by repeating earlier requests.

\begin{figure}[t]
\centering
\begin{tikzpicture}[x=1cm,y=1cm,
  live/.style={circle,fill=black!65,inner sep=0pt,minimum size=2.7pt},
  dead/.style={circle,fill=black!15,inner sep=0pt,minimum size=2.7pt},
  oval/.style={draw=black!45,densely dashed},
  deadoval/.style={draw=black!18,densely dashed},
  alg/.style={circle,fill=blue!70!black,inner sep=0pt,minimum size=6.5pt},
  oldalg/.style={circle,draw=blue!70!black,line width=0.8pt,inner sep=0pt,minimum size=6.5pt},
  pics/cl/.style={code={
    \foreach \dx/\dy in {-0.24/0.05,-0.09/0.17,0.1/0.13,0.26/-0.01,0.04/-0.14,-0.16/-0.12,0.0/0.0}
      \node[#1] at (\dx,\dy) {};
  }},
  pics/half/.style={code={
    \foreach \dx/\dy in {-0.24/0.05,-0.09/0.17,0.1/0.13}
      \node[live] at (\dx,\dy) {};
    \foreach \dx/\dy in {0.26/-0.01,0.04/-0.14,-0.16/-0.12,0.0/0.0}
      \node[dead] at (\dx,\dy) {};
  }}]
  \draw[black!30,rounded corners=10pt] (0,0) rectangle (6.3,3.3);
  \node[font=\small,anchor=north west] at (0.08,3.25) {$Q_t$};
  \foreach \x/\y in {1.0/2.35, 2.45/0.95, 3.35/2.45, 5.25/2.4, 5.0/0.85, 0.95/0.9} {
    \draw[oval] (\x,\y) ellipse (0.44 and 0.32);
    \pic at (\x,\y) {cl=live};
  }
  \node at (3.45,2.58) {$\star$};
  \node[font=\small,anchor=south] at (3.45,2.78) {$\sigma$};
  \node[alg] at (2.45,0.95) {};
  \node[font=\small,anchor=north,blue!70!black] at (2.45,0.58) {$Y_t$};
  \draw[{Stealth[length=1.8mm]}-{Stealth[length=1.8mm]},black!60] (2.95,0.93) -- (4.5,0.87)
    node[midway,below,font=\small,black] {$>0.77k$};
  \node[font=\small] at (3.15,-0.35) {(a) before $r_{t+1}$};
  \begin{scope}[xshift=7.1cm]
  \draw[black!30,rounded corners=10pt] (0,0) rectangle (6.3,3.3);
  \node[font=\small,anchor=north west] at (0.08,3.25) {$Q_{t+1}$};
  \foreach \x/\y in {2.45/0.95, 5.0/0.85, 0.95/0.9} {
    \draw[deadoval] (\x,\y) ellipse (0.44 and 0.32);
    \pic at (\x,\y) {cl=dead};
  }
  \foreach \x/\y in {1.0/2.35, 3.35/2.45, 5.25/2.4} {
    \draw[oval] (\x,\y) ellipse (0.44 and 0.32);
    \pic at (\x,\y) {half};
  }
  \node at (3.45,2.58) {$\star$};
  \node[font=\small,anchor=south] at (3.45,2.78) {$\sigma$};
  \node[oldalg] at (2.45,0.95) {};
  \node[alg] at (5.01,2.45) {};
  \draw[blue!70!black,line width=1pt,-{Stealth[length=2.2mm]}] (2.6,1.1) .. controls (3.4,1.9) and (4.2,2.1) .. (4.88,2.4)
    node[pos=0.5,below right,font=\small,black] {$>0.77k$};
  \node[font=\small,anchor=north,blue!70!black] at (2.45,0.58) {$Y_t$};
  \node[font=\small,anchor=south,blue!70!black] at (5.01,2.62) {$Y_{t+1}$};
  \node[font=\small] at (3.15,-0.35) {(b) after $r_{t+1}$ and the confinement requests};
  \end{scope}
\end{tikzpicture}
\caption{One step of the adversary of \cref{thm:lower}, schematically.
Dots are configurations that serve the setup requests, dashed ovals are their components, at distance more than $0.77k$ from each other, and $\star$ marks the hidden permutation $\sigma$; $Q_t$ (dark dots) is the set of feasible configurations after the requests $r_1,\dots,r_t$.
(a) With probability close to~$1$, the configuration $Y_t$ of the algorithm lies in $Q_t$ and in a component without $\sigma$, since, given the requests so far, $\sigma$ is uniform on $Q_t$ and independent of $Y_t$, components have fewer than $2^k$ configurations, and the expected inverse of $|Q_t|$ is far below $2^{-k}$.
(b) With probability more than $0.69/k$, the request $r_{t+1}$ forbids a position that all configurations of the component of $Y_t$ use and $\sigma$ avoids, which eliminates the component (light dots fail $r_{t+1}$), and the confinement requests, which repeat earlier requests, drive the algorithm into $Q_{t+1}$, a distance of more than $0.77k$ away.}
\label{fig:clusters}
\end{figure}
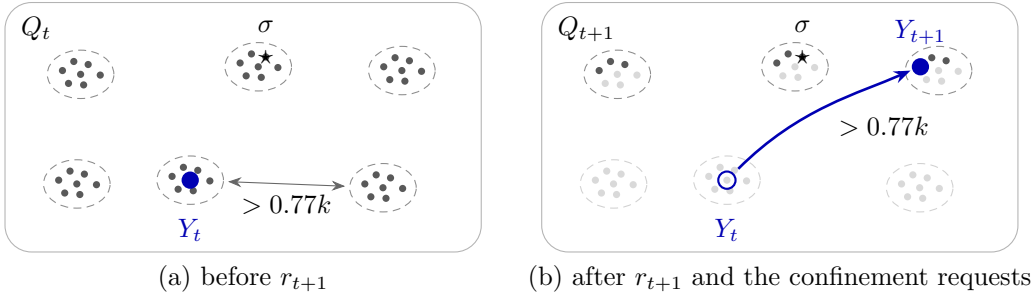

\paragraph{The algorithm for weighted uniform metrics.}
Bijoy, Mondal, and Chiplunkar define phases recursively.
The \emph{witness} of a phase of level~$1$ is its first request that the heavier servers, held at their points, do not serve; server~$1$ then serves it, and the phase lasts while these servers serve the requests.
A phase of level $\ell\ge2$ starts with an \emph{explore part} made of phases of level $\ell-1$, after which the $d_\ell-1$ points of $M_\ell$ named most often by the witnesses of the explore part form the \emph{critical set}, with $d_\ell=2^{5^{\ell-1}-1}$.
In its \emph{exploit part}, phases of level $\ell-1$ are run in parallel, one run for each critical point, with server~$\ell$ held at that point (\cref{fig:phase}).
An offline solution that moves server~$\ell$ pays its weight.
One that keeps server~$\ell$ still pays in many of the phases of level $\ell-1$, unless its heavier servers stand at points that many witnesses name: outside the critical set, server~$\ell$ serves few of the witnesses of the explore part, and at a critical point~$p$, it serves none of the witnesses of the run for~$p$.
At level $k$ there are no heavier servers, so once the weights are rounded so that consecutive ratios are integers at least~$2$, every offline solution pays at least a $2^{-k}$ fraction of the largest weight per completed phase of level $k$.

Their algorithm, like ours, simulates all runs of the exploit part and lets the real servers follow one run at a time.
It follows a uniformly random run until that run ends, then another one, and so on, so its cost bound for a phase of level $\ell$ is about $1+H_{d_\ell-1}$ times that for level $\ell-1$, where $H_n\approx\ln n$ is the $n$-th harmonic number; these factors multiply out to $\exp(O(k^2))$.
We instead draw the followed run with exponential weights on the costs that the runs have paid so far, delete a run when it ends, and change the followed run by a maximal coupling of consecutive distributions.
If the configurations of the $n$ runs and the starting configuration of the real servers lie in a set of diameter $\Delta$, the real servers move at most $\Delta+2\Delta\ln n$ plus twice the cost of any run that has not ended, in expectation (\cref{lem:aggregate}).

After rounding, the explore part and every run consist of $w_\ell/w_{\ell-1}$ phases of level $\ell-1$; so if such a phase costs at most $c_{\ell-1}$ times the weight of server~$\ell-1$, each of them costs at most $c_{\ell-1}$ times the weight of server~$\ell$, and $\Delta$ is at most twice that weight.
So the logarithm of the number of runs enters additively: a phase of level $\ell$ costs at most $c_\ell$ times the weight of server~$\ell$, where $c_\ell=3c_{\ell-1}+4\ln(d_\ell-1)+2$ collects $c_{\ell-1}$ for the explore part, $2c_{\ell-1}$ for the followed runs, and $2+4\ln(d_\ell-1)$ for switching.
Hence $c_k\le10\cdot5^{k-1}$, and with the factors $2^k$ from the offline bound and $2^{k-1}$ from rounding the weights, the ratio is $20^k$.
We also write out the phases and the offline lower bound of~\cite{BMC26} for generalized requests, for which the authors sketch the proofs, with the requests counted separately in each coordinate.

\begin{figure}[t]
\centering
\begin{tikzpicture}[x=1cm,y=1cm,
  lane/.style={black!40,line width=0.8pt},
  tick/.style={black!60,line width=0.6pt},
  master/.style={blue!70!black,line width=2.4pt,line cap=round},
  jump/.style={blue!70!black,line width=1.1pt,-{Stealth[length=2.2mm]}},
  pbox/.style={draw=black!45,fill=black!5,minimum height=4.6mm,minimum width=10.5mm,inner sep=0pt},
  bar/.style={fill=black!15,draw=black!40,line width=0.3pt},
  topbar/.style={fill=black!60,draw=black!60,line width=0.3pt},
  brace/.style={decorate,decoration={brace,amplitude=4pt,mirror},draw=black!60}]
  \foreach \i in {0,1,2,3} \node[pbox] at (0.65+1.15*\i,1.5) {};
  \draw[master] (0.1,1.5) -- (4.65,1.5);
  \draw[brace] (0.1,1.08) -- (4.65,1.08)
    node[midway,below=5pt,font=\small,align=center] {explore part:\\phases of level $\ell-1$};
  \foreach \i/\h in {0/0.55,1/0.18,2/0.78,3/0.32,4/0.66,5/0.12,6/0.48}
    \path[bar] (0.95+0.42*\i,2.2) rectangle ++(0.3,\h);
  \foreach \i/\h in {0/0.55,2/0.78,4/0.66,6/0.48}
    \path[topbar] (0.95+0.42*\i,2.2) rectangle ++(0.3,\h);
  \draw[black!50] (0.85,2.2) -- (3.9,2.2);
  \node[font=\small,anchor=south,align=center] at (2.4,3.02) {witnesses per point of $M_\ell$;\\the top $d_\ell-1$ form $S$};
  \draw[densely dashed,black!45] (4.95,0.55) -- (4.95,3.3);
  \foreach \y/\e/\lab in {2.75/8.2/p_1, 2.15/12.7/p_2, 1.5/9.6/p_3, 0.85/11.0/p_4} {
    \draw[lane] (5.6,\y) -- (\e,\y);
    \node[font=\small,black!80] at (\e,\y) {$\times$};
    \node[font=\small,anchor=south west,inner sep=1pt] at (5.05,\y+0.02) {$\lab$};
  }
  \foreach \x/\y in {6.0/2.75, 7.6/2.75,
                     6.4/2.15, 9.9/2.15, 11.7/2.15,
                     6.2/1.5, 6.45/1.5, 6.7/1.5, 6.95/1.5, 8.8/1.5,
                     6.8/0.85, 8.9/0.85, 10.4/0.85}
    \draw[tick] (\x,\y-0.11) -- (\x,\y+0.11);
  \draw[master] (4.65,1.5) -- (6.95,1.5);
  \draw[jump] (6.95,1.5) -- (6.95,2.65);
  \draw[master] (6.95,2.75) -- (8.2,2.75);
  \draw[jump] (8.2,2.75) -- (8.2,2.25);
  \draw[master] (8.2,2.15) -- (12.55,2.15);
  \draw[densely dashed,black!45] (12.95,0.55) -- (12.95,3.3);
  \node[font=\small,anchor=west,align=left] at (13.05,2.15) {phase\\stops};
  \draw[brace] (5.05,0.5) -- (12.85,0.5)
    node[midway,below=5pt,font=\small,align=center] {exploit part: for each $p\in S$, a copy keeps server $\ell$ at $p$\\and runs phases of level $\ell-1$ until they stop ($\times$)};
\end{tikzpicture}
\caption{A phase of level $\ell$, after~\cite{BMC26}, and the strategy of \cref{lem:strategy}, drawn with four copies.
The real servers (thick line) serve the explore part, and the $d_\ell-1$ points of $M_\ell$ named most often by its witnesses (bars) form the critical set $S$.
In the exploit part, a simulated copy for each $p\in S$ keeps server~$\ell$ at $p$ and reads the same requests; ticks mark the moves of the copies, and a copy is deleted when its phases stop ($\times$).
The real servers follow one copy at a time, drawn with weights that decay exponentially with the costs of the copies.
The followed copy changes by a maximal coupling of consecutive distributions, so the real servers switch when the followed copy stops and possibly when it moves.}
\label{fig:phase}
\end{figure}
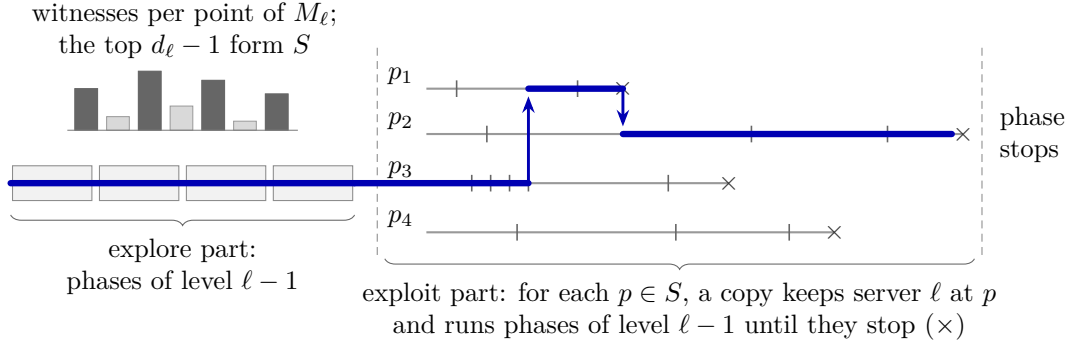

\subsection{Related work}

The generalized $k$-server problem on uniform metrics is a metrical service system~\cite{CL91} on the product space, whose request sets are the sets of configurations that serve a request.
For the $k$-server problem, the work function algorithm is $(2k-1)$-competitive~\cite{KP95}, and a recent preprint of Coester, Koutsoupias, and Zbysi\'nski proves that it is $k$-competitive~\cite{CKZ26}.
For randomized $k$-server, mirror descent with a multiscale entropy gives $O(\log^2k)$ on hierarchically separated trees~\cite{BCLLM18}, and Bubeck, Coester, and Rabani~\cite{BCR23} showed that $\Omega(\log^2k)$ is needed on some metric spaces.

\section{Preliminaries}\label{sec:prelim}

Fix $k\ge1$ and pairwise disjoint uniform metric spaces $M_1,\dots,M_k$, where distinct points of $M_i$ are at distance $w_i\ge0$ and $w_1\le\dots\le w_k$.
In \cref{sec:upper,sec:lower} all weights are $1$.
A \emph{configuration} is a tuple $x$ in $\cX:=M_1\times\dots\times M_k$, whose coordinate $x_i$, also written $x[i]$, is the position of server $i$.
Moving the servers from $x$ to $y$ costs
\[
d_w(x,y)=\sum_{i=1}^kw_i\,\one[x_i\ne y_i];
\]
with unit weights this is the Hamming distance $d(x,y)=|\{i\in[k]:x_i\ne y_i\}|\le k$.
A request is a tuple $r\in\cX$, and a configuration $x$ \emph{serves} $r$ if $x_i=r_i$ for some $i$.
Let $G(r)$ be the set of configurations that serve~$r$.
Weighted $k$-server on a uniform metric $M'$ is the special case in which every $M_i$ is a copy of $M'$ and every request consists of the $k$ copies of one point of $M'$~\cite{BMC26}.

An online algorithm starts in a configuration $a$.
When request $r_t$ arrives, it moves to a configuration $X_t\in G(r_t)$, possibly $X_t=X_{t-1}$, and pays $d_w(X_{t-1},X_t)$.
An offline solution for $\rho=(r_1,\dots,r_T)$ is a sequence $o_0,o_1,\dots,o_T$ of configurations with $o_t\in G(r_t)$ for $t\ge1$.
Its cost is $\sum_{t=1}^Td_w(o_{t-1},o_t)$, and $\OPT(\rho)$ is the minimum cost of an offline solution with $o_0=a$.
A randomized online algorithm is \emph{$\gamma$-competitive} if there is a constant $c$, which may depend on the spaces, the weights and $a$, such that $\E[\ALG(\rho)]\le\gamma\cdot\OPT(\rho)+c$ for every finite request sequence $\rho$.
The expectation is over the random choices of the algorithm, and $\rho$ is fixed in advance; that is, the adversary is oblivious~\cite{BBKTW94}.

\paragraph{Repairs and lazy solutions.}
For a configuration $x$, a request $r$, and $i\in[k]$, let $R_i(x,r)$ be the configuration obtained from $x$ by moving server $i$ to $r_i$; that is, $R_i(x,r)_i=r_i$ and $R_i(x,r)_\ell=x_\ell$ for $\ell\ne i$.
If $x\notin G(r)$, then every $R_i(x,r)$ serves $r$ and $d(x,R_i(x,r))=1$; these $k$ configurations are the \emph{repairs} of $x$ for $r$.
An offline solution $\tilde o_0,\tilde o_1,\dots,\tilde o_T$ is \emph{lazy} if, for every $t$, it stays ($\tilde o_t=\tilde o_{t-1}$) when $\tilde o_{t-1}\in G(r_t)$, and otherwise moves to a repair of $\tilde o_{t-1}$ for~$r_t$.

\begin{lemma}\label[lemma]{lem:lazy}
With unit weights, let $o_0,\dots,o_T$ be an offline solution of cost $C$, and let $a$ be a configuration.
There is a lazy offline solution with $\tilde o_0=a$ that moves at most $C+d(a,o_0)$ times.
\end{lemma}

\begin{proof}
Let $\tilde o_0=a$.
For $t\ge1$, if $\tilde o_{t-1}\in G(r_t)$, let $\tilde o_t=\tilde o_{t-1}$.
Otherwise choose $i$ with $(o_t)_i=(r_t)_i$, which exists because $o_t$ serves $r_t$, and let $\tilde o_t=R_i(\tilde o_{t-1},r_t)$.
Let $m_t\in\{0,1\}$ indicate whether the lazy solution moves at time $t$.
If $m_t=1$, then $(\tilde o_{t-1})_i\ne(r_t)_i=(o_t)_i=(\tilde o_t)_i$, and $\tilde o_t$ agrees with $\tilde o_{t-1}$ elsewhere, so $d(\tilde o_t,o_t)=d(\tilde o_{t-1},o_t)-1$.
Hence in both cases
\[
m_t+d(\tilde o_t,o_t)=d(\tilde o_{t-1},o_t)\le d(\tilde o_{t-1},o_{t-1})+d(o_{t-1},o_t).
\]
Summing over $t$ gives $\sum_tm_t\le d(\tilde o_0,o_0)+C-d(\tilde o_T,o_T)\le d(a,o_0)+C$.
\end{proof}

\section{An \texorpdfstring{$O(k\log k)$}{O(k log k)}-competitive algorithm for unit weights}\label{sec:upper}

\subsection{The weight update}

For a function $\phi:\cX\to[0,\infty)$ with finite support, write $\|\phi\|=\sum_x\phi(x)$ and $\phi(A)=\sum_{x\in A}\phi(x)$.
For a request $r$, define $\cT_r\phi:\cX\to[0,\infty)$ by
\begin{equation}\label{eq:update}
(\cT_r\phi)(y)=\one[y\in G(r)]\,\phi(y)+\frac1{2k}\sum_{x\notin G(r)}\ \sum_{i=1}^k\phi(x)\,\one[y=R_i(x,r)].
\end{equation}
Thus $\cT_r\phi$ keeps the weight of every configuration that serves $r$, and every configuration $x$ that fails $r$ passes $\phi(x)/(2k)$ to each of its $k$ repairs.
The map $\cT_r$ is linear, $\cT_r\phi$ is supported on $G(r)$, and
\begin{equation}\label{eq:mass}
\|\cT_r\phi\|=\phi(G(r))+\tfrac12\,\phi(\cX\setminus G(r)).
\end{equation}

Given the starting configuration $a$ and the requests $r_1,r_2,\dots$, let $\phi_0=\one_{\{a\}}$, $\phi_t=\cT_{r_t}\phi_{t-1}$, $\Phi_t=\|\phi_t\|$, and $\mu_t=\phi_t/\Phi_t$.
By~\eqref{eq:mass}, $\Phi_t\ge\Phi_{t-1}/2>0$.
By linearity, $\mu_t=\cT_{r_t}\mu_{t-1}/z_t$, where
\begin{equation}\label{eq:z}
z_t:=\frac{\Phi_t}{\Phi_{t-1}}=1-\frac{b_t}{2},\qquad b_t:=\mu_{t-1}(\cX\setminus G(r_t)).
\end{equation}
The support of $\phi_t$ lies in $\prod_{i=1}^k\bigl(\{a_i\}\cup\{(r_s)_i:s\le t\}\bigr)$, which is finite even if the spaces are infinite.
\Cref{fig:update} shows one update.

\subsection{The algorithm}

The algorithm computes $\mu_t$ from the requests seen so far, starts at $X_0=a$, and moves its servers as follows when request $r_t$ arrives.
\begin{enumerate}[(1)]
\item If $X_{t-1}\in G(r_t)$, it stays: $X_t=X_{t-1}$.
\item Otherwise, with probability $1/2$ it moves to $R_i(X_{t-1},r_t)$ for a uniformly random $i\in[k]$, and with probability $1/2$ it moves to a configuration drawn from $\mu_t$, independently of everything else.
\end{enumerate}
In every case $X_t\in G(r_t)$, because $\mu_t$ is supported on $G(r_t)$.

\begin{lemma}\label[lemma]{lem:follow}
For every fixed request sequence and every $t\ge0$, the configuration $X_t$ has distribution~$\mu_t$.
\end{lemma}

\begin{proof}
We use induction on $t$; the case $t=0$ holds because $X_0=a$.
Suppose that $X_{t-1}$ has distribution $\mu_{t-1}$, and write $r=r_t$.
Case~(1) and the repair branch of case~(2) together put probability
\[
\one[y\in G(r)]\,\mu_{t-1}(y)+\sum_{x\notin G(r)}\mu_{t-1}(x)\cdot\frac12\cdot\frac1k\sum_{i=1}^k\one[y=R_i(x,r)]=(\cT_r\mu_{t-1})(y)=z_t\,\mu_t(y)
\]
on each configuration $y$.
The resampling branch of case~(2) is taken with total probability $b_t/2$, and then lands at $y$ with probability $\mu_t(y)$.
So $\Pr[X_t=y]=(z_t+b_t/2)\,\mu_t(y)=\mu_t(y)$ by~\eqref{eq:z}.
\end{proof}

\begin{lemma}\label[lemma]{lem:movement}
For every fixed request sequence and every $t\ge1$, $\E[d(X_{t-1},X_t)]\le(k+1)\ln(\Phi_{t-1}/\Phi_t)$.
\end{lemma}

\begin{proof}
The servers move only in case~(2), which by \cref{lem:follow} has probability $b_t$.
In that case the repair branch costs $1$ and the resampling branch costs at most $k$.
Hence $\E[d(X_{t-1},X_t)]\le b_t(\frac12+\frac k2)=(k+1)(1-z_t)\le(k+1)\ln(1/z_t)$, using $1-z\le-\ln z$ for $z\in(0,1]$.
\end{proof}

\begin{lemma}\label[lemma]{lem:weight}
If $\tilde o_0=a,\tilde o_1,\dots,\tilde o_T$ is a lazy offline solution that moves $m$ times, then $\Phi_T\ge\phi_T(\tilde o_T)\ge(2k)^{-m}$.
\end{lemma}

\begin{proof}
All terms in~\eqref{eq:update} are nonnegative.
So if $\tilde o_t=\tilde o_{t-1}\in G(r_t)$, then $\phi_t(\tilde o_t)\ge\phi_{t-1}(\tilde o_{t-1})$, and if $\tilde o_t$ is a repair of $\tilde o_{t-1}\notin G(r_t)$, then $\phi_t(\tilde o_t)\ge\phi_{t-1}(\tilde o_{t-1})/(2k)$.
Since $\phi_0(a)=1$, the claim follows by induction on $t$.
\end{proof}

\begin{proof}[Proof of \cref{thm:upper}]
Fix a finite request sequence $\rho=(r_1,\dots,r_T)$.
Summing \cref{lem:movement} over $t$ and using $\Phi_0=1$ gives $\E[\ALG(\rho)]\le(k+1)\ln(1/\Phi_T)$.
Apply \cref{lem:lazy} to an optimal offline solution starting from $a$: there is a lazy solution from $a$ that moves at most $\OPT(\rho)$ times.
By \cref{lem:weight}, $\Phi_T\ge(2k)^{-\OPT(\rho)}$.
Hence $\E[\ALG(\rho)]\le(k+1)\ln(2k)\cdot\OPT(\rho)$.
\end{proof}

\begin{remark}
The analysis charges the movement of the algorithm to the moves of a single lazy offline solution over the whole sequence, rather than phase by phase.
Against an offline solution that starts from another configuration $b$, the same argument gives $\E[\ALG(\rho)]\le(k+1)\ln(2k)\,(C+d(a,b))$, where $C$ is the cost of that solution, so the additive term is at most $k(k+1)\ln(2k)$.
\Cref{lem:follow} uses that the request sequence is fixed in advance; the guarantee is against an oblivious adversary.
Passing a fraction $\theta/k$ of the weight to each repair, for some $\theta\in(0,1)$, and taking the repair branch with probability $\theta$ gives the ratio $(k+\frac{\theta}{1-\theta})\ln\frac k\theta$ in the same way; for $k\ge2$ and $\frac{\theta}{1-\theta}=\sqrt{k/\ln k}$, this is at most $k\ln k+2\sqrt{k\ln k}+1$.
\end{remark}

\subsection{Requests that specify sets}

Suppose now that a request specifies a finite set $P_i\subseteq M_i$ for every $i\in[k]$, with $1\le\sum_i|P_i|\le B$, and that a configuration $x$ serves it if $x_i\in P_i$ for some $i$.
The repairs of a configuration $x$ that fails the request are the $\sum_i|P_i|$ configurations obtained from $x$ by moving one server $i$ to one point of $P_i$; they are distinct, at distance $1$ from $x$, and serve the request.
We run the algorithm of \cref{sec:upper} with these repairs: a configuration that fails a request passes the same fraction $1/(2\sum_i|P_i|)$ of its weight to each of its repairs, and the repair branch of case~(2) moves to a uniformly random repair.

\begin{theorem}\label{thm:sets}
For requests that specify sets of at most $B$ points in total, this algorithm satisfies $\E[\ALG(\rho)]\le(k+1)\ln(2B)\cdot\OPT(\rho)$ for every finite request sequence $\rho$.
\end{theorem}

\begin{proof}
\Cref{lem:follow,lem:movement} hold with the same proofs.
In the proof of \cref{lem:lazy}, choose $i$ with $(o_t)_i\in P_i$ and move server $i$ of the lazy solution to $(o_t)_i$; since the lazy solution fails the request, this decreases its distance to $o_t$ by one, so the lemma holds.
Each move of a lazy solution keeps at least a $1/(2B)$ fraction of its weight, so \cref{lem:weight} holds with $2B$ in place of $2k$, and the proof of \cref{thm:upper} gives the bound.
\end{proof}

In paging with heterogeneous cache slots~\cite{CHLPRSY23}, a cache has $k$ slots, initially empty, each request names a page $p$ and a set $S$ of slots, and it is served if a slot in $S$ holds $p$; loading a page into a slot costs $1$, and emptying a slot is free.
Taking every $M_i$ to be the set of pages together with a point for an empty slot, $P_i=\{p\}$ for $i\in S$, and $P_i=\emptyset$ otherwise gives $B\le k$; since an offline solution never needs to empty a slot, the costs agree, and \cref{thm:sets} gives the ratio $(k+1)\ln(2k)$.
The previous bound was $O(k^2\log k)$, through a reduction to generalized $k$-server~\cite{CHLPRSY23,BJS19}, and the lower bound is $\Omega(k)$~\cite{CHLPRSY23}.
In metrical service systems with $k$ servers on a uniform metric~\cite{CV15}, a request is a set $R$ of at most $\ell$ points, served if some server is at a point of $R$.
Labeling the servers changes no cost, and taking $P_i=R$ for every $i$ gives $B\le k\ell$ and the ratio $(k+1)\ln(2k\ell)$; for $\ell\ge2$ the previous bound was $O(k^3\log\ell)$, and the lower bound is $\Omega(\log(k\ell))$~\cite{CV15}.

\section{An \texorpdfstring{$\Omega(k\log k)$}{Omega(k log k)} lower bound for unit weights}\label{sec:lower}

The proof follows the plan in \cref{sec:techniques}.
\Cref{sec:perm,sec:moments} set up random permutation constraints and compute two moments of their number of solutions, and \cref{sec:components} shows that for some fixed system of constraints, the solutions are numerous and split into components at distance more than $\beta k$ from each other, each of which varies in fewer than $\eta k$ coordinates and has fewer than $2^k$ configurations (\cref{lem:fixed,lem:structure}).
\Cref{sec:elim} shows that a request that forbids one position to one server, drawn among those that the hidden permutation serves, eliminates a given component without it with probability at least $(\beta-\eta)/(k-1)$; \cref{sec:posterior} shows that the hidden permutation is uniform on the feasible configurations and that the expected inverse of their number stays small, and \cref{sec:adversary} assembles the adversary.

Throughout this section $k\ge2$, and every space $M_i$ is $[k]=\{1,\dots,k\}$ with unit distances, so $\cX=[k]^k$.
Statements ``for large $k$'' hold for all $k\ge k_0$, for an absolute constant $k_0$.
We fix the constants
\[
\varepsilon=\tfrac1{10},\qquad\alpha=1-\varepsilon,\qquad\delta=\tfrac1{25},\qquad\eta=\tfrac2{25},\qquad\beta=\tfrac{77}{100},\qquad h=\lfloor\alpha\ln k\rfloor,
\]
and write $N=k!$ and $p=k!/k^k$.
A configuration $x$ is a \emph{permutation} if $x_1,\dots,x_k$ are distinct.
Let $\Perm\subseteq\cX$ be the set of permutations, so $|\Perm|=N$.
For $x\in\cX$ and $I\subseteq[k]$, write $x(I)=\{x_\ell:\ell\in I\}$.
We use the bounds $(k/e)^k\le k!\le ek(k/e)^k$, which give
\begin{equation}\label{eq:stirling}
k\ln k-k\le\ln N\le k\ln k-k+\ln k+1,\qquad -k\le\ln p\le-k+\ln k+1.
\end{equation}
We also use $\alpha\ln k-1<h\le\alpha\ln k$.

\subsection{Permutation constraints}\label{sec:perm}

A \emph{family} is a $k$-tuple $\pi=(\pi_1,\dots,\pi_k)$ of bijections of $[k]$.
For $x\in\cX$ let $\pi(x)=(\pi_1(x_1),\dots,\pi_k(x_k))$.
The family \emph{accepts} a permutation $x$ if $\pi(x)$ is a permutation.
A \emph{system} is an $h$-tuple $F=(\pi^1,\dots,\pi^h)$ of families.
Let $Q(F)$ be the set of permutations that every family of $F$ accepts, and let $Z(F)=|Q(F)|$.

\begin{lemma}\label[lemma]{lem:enforce}
A configuration serves all of the $k+hk$ requests
\[
(a,a,\dots,a)\quad(a\in[k]),\qquad r^{j,a}=\bigl((\pi^j_1)^{-1}(a),\dots,(\pi^j_k)^{-1}(a)\bigr)\quad(j\in[h],\ a\in[k])
\]
if and only if it lies in $Q(F)$.
\end{lemma}

\begin{proof}
A configuration $x$ serves $(a,\dots,a)$ if and only if $a\in\{x_1,\dots,x_k\}$, so $x$ serves all these requests if and only if $x\in\Perm$.
It serves $r^{j,a}$ if and only if $\pi^j_i(x_i)=a$ for some $i$, so it serves $r^{j,a}$ for all $a$ if and only if $\pi^j(x)$ takes every value in $[k]$, that is, $\pi^j(x)\in\Perm$.
\end{proof}

We call these requests the \emph{setup requests} of $F$.

\subsection{Two moments}\label{sec:moments}

Let $\Pr_0$ and $\E_0$ refer to a uniformly random system, in which all $hk$ bijections $\pi^j_i$ are independent and uniform.

\begin{lemma}\label[lemma]{lem:first}
$\E_0Z=Np^h$, and for large $k$, $\ \varepsilon k\ln k-k\le\ln\E_0Z\le\varepsilon k\ln k+k$.
\end{lemma}

\begin{proof}
For a fixed $x\in\Perm$, the values $\pi^j_1(x_1),\dots,\pi^j_k(x_k)$ are independent and uniform on $[k]$, so $\pi^j$ accepts $x$ with probability $p$, independently over $j$.
Hence $\E_0Z=Np^h$.
By~\eqref{eq:stirling} and $h\le\alpha\ln k$, $\ln\E_0Z=\ln N+h\ln p\ge k\ln k-k-\alpha k\ln k=\varepsilon k\ln k-k$.
By~\eqref{eq:stirling} and $\alpha\ln k-1<h\le\alpha\ln k$,
\[
\ln\E_0Z\le k\ln k-k+\ln k+1-(\alpha\ln k-1)k+\alpha\ln k\,(\ln k+1)=\varepsilon k\ln k+\ln k+1+\alpha\ln k\,(\ln k+1),
\]
which is at most $\varepsilon k\ln k+k$ for large $k$.
\end{proof}

\begin{lemma}\label[lemma]{lem:residual}
Fix $x\in\Perm$, and let $\pi$ be a uniformly random family conditioned on accepting $x$.
Then $v=\pi(x)$ is a uniformly random permutation.
Conditioned on $v$, the restrictions of $\pi_1,\dots,\pi_k$ to $[k]\setminus\{x_1\},\dots,[k]\setminus\{x_k\}$ are independent, and the restriction of $\pi_i$ is a uniformly random bijection from $[k]\setminus\{x_i\}$ onto $[k]\setminus\{v_i\}$.
\end{lemma}

\begin{proof}
For each $v\in\Perm$, the families $\pi$ with $\pi(x)=v$ are obtained by choosing, independently for each $i$, a bijection of $[k]$ that maps $x_i$ to $v_i$.
There are $((k-1)!)^k$ of them, a number that does not depend on $v$.
\end{proof}

\begin{lemma}\label[lemma]{lem:pair}
Let $x,y\in\Perm$ differ in exactly $d$ coordinates, and let $D_d$ be the number of derangements of $d$ elements ($D_0=1$, $D_1=0$).
For a uniformly random family $\pi$, the probability that $\pi$ accepts $y$, conditioned on $\pi$ accepting $x$, is $D_d/(k-1)^d$.
\end{lemma}

\begin{proof}
Let $J=\{i:y_i\ne x_i\}$, and use \cref{lem:residual}.
For $i\notin J$ we have $\pi_i(y_i)=v_i$.
For $i\in J$, the value $\pi_i(y_i)$ is uniform on $[k]\setminus\{v_i\}$, independently over $i\in J$.
Since $\pi(y)$ agrees with the permutation $v$ outside $J$, it is a permutation if and only if the values $\pi_i(y_i)$, $i\in J$, are distinct and form the set $v(J)$.
Writing $\pi_i(y_i)=v_{f(i)}$, this says that $f$ is a bijection of $J$ with $f(i)\ne i$ for all $i\in J$, that is, a derangement of $J$.
Each of the $D_d$ derangements has probability $(k-1)^{-d}$.
\end{proof}

\begin{lemma}\label[lemma]{lem:second}
For $0\le d\le k$, let $Z_d$ be the number of ordered pairs of configurations in $Q(F)$ that differ in exactly $d$ coordinates, and let
\[
s_d=\binom kdD_d\Bigl(\frac{D_d}{(k-1)^d}\Bigr)^h .
\]
Then $\E_0Z_d=s_d\,\E_0Z$; for large $k$ and $2\le d\le k$, we have $\ln s_d\le k\ln k\cdot g(d/k)+3k$, where $g(u)=u(\varepsilon+\alpha\ln u)$, and $\E_0[Z^2]\le e^{5k}(\E_0Z)^2$.
\end{lemma}

\begin{proof}
For a fixed $x\in\Perm$, exactly $\binom kdD_d$ permutations differ from $x$ in $d$ coordinates, and by \cref{lem:pair} and the independence of the families, each of them lies in $Q(F)$ together with $x$ with probability $p^h(D_d/(k-1)^d)^h$.
Hence $\E_0Z_d=Np^hs_d=s_d\,\E_0Z$.
Here $s_0=1$ and $s_1=0$.
Let $2\le d\le k$.
We use $\binom kdD_d\le k^d$, $D_d\le d!\le ed(d/e)^d$, and $(k-1)^d\ge k^d(1-1/k)^k\ge k^d/4$, which give
\[
\ln s_d\le d\ln k+h\bigl(\ln(4ed)+\Lambda_d\bigr),\qquad\Lambda_d:=d\ln(d/k)-d<0.
\]
Since $h>\alpha\ln k-1$ and $\Lambda_d<0$, we have $h\Lambda_d\le\alpha\ln k\cdot\Lambda_d-\Lambda_d\le\alpha\ln k\cdot\Lambda_d+2k$, because $-\Lambda_d=d\ln(k/d)+d\le k/e+k$.
Also $h\ln(4ed)\le\alpha\ln k\cdot\ln(4ek)\le k$ for large $k$.
With $u=d/k$ this gives
\[
\ln s_d\le k\ln k\,\bigl(u+\alpha u\ln u-\alpha u\bigr)+3k=k\ln k\cdot g(u)+3k.
\]
The function $g$ is convex on $(0,1]$, since $g''(u)=\alpha/u>0$; it tends to $0$ as $u\to0$, and $g(1)=\varepsilon$.
Hence $g\le\varepsilon$ on $(0,1]$, and $s_d\le e^{\varepsilon k\ln k+3k}$.
Since $Z^2=\sum_dZ_d$, \cref{lem:first} gives
\[
\frac{\E_0[Z^2]}{(\E_0Z)^2}=\frac{\sum_ds_d}{\E_0Z}\le\frac{1+k\,e^{\varepsilon k\ln k+3k}}{e^{\varepsilon k\ln k-k}}\le2k\,e^{4k}\le e^{5k}.\qedhere
\]
\end{proof}

\subsection{Components of the feasible set}\label{sec:components}

For a system $F$, join two configurations of $Q(F)$ by an edge if they differ in fewer than $\delta k$ coordinates.
The connected components of this graph are the \emph{components} of $Q(F)$.
For a component $C$ and $i\in[k]$, let $A_i(C)=\{y_i:y\in C\}$ be the set of positions that server $i$ takes in~$C$; the component \emph{varies} in coordinate $i$ if $|A_i(C)|\ge2$.
To bound these sets, we fix $x\in C$ and a set $I$ of coordinates and consider the pairs $(i,y_i)$ with $i\in I$ and $y_i\ne x_i$, for the configurations $y\in Q(F)$ that agree with $x$ outside $I$; by \cref{lem:alt}, the following set contains them.

\begin{definition}
For $x\in Q(F)$ and $I\subseteq[k]$, let $\cL_x(I)$ be the set of pairs $(i,b)$ with $i\in I$ and $b\in x(I)\setminus\{x_i\}$ such that $\pi^j_i(b)\in\pi^j(x)(I)$ for every $j\in[h]$.
\end{definition}

\begin{lemma}\label[lemma]{lem:alt}
Let $x,y\in Q(F)$ agree outside $I\subseteq[k]$.
Then $(i,y_i)\in\cL_x(I)$ for every $i\in I$ with $y_i\ne x_i$.
\end{lemma}

\begin{proof}
Since $x$ and $y$ are permutations that agree outside $I$, we have $y(I)=[k]\setminus x([k]\setminus I)=x(I)$, so $y_i\in x(I)\setminus\{x_i\}$.
For each $j$, the permutations $\pi^j(x)$ and $\pi^j(y)$ also agree outside $I$, so likewise $\pi^j(y)(I)=\pi^j(x)(I)$, and this set contains $\pi^j_i(y_i)$.
\end{proof}

So a component can vary in many coordinates or take many values only if $\cL_x(I)$ is large for some $x$ and some set $I$ (\cref{lem:structure}).
The system $F$ is \emph{good} if it has the following three properties.
\begin{enumerate}[label=(G\arabic*),nosep]
\item No two configurations of $Q(F)$ differ in at least $\delta k$ and at most $\beta k$ coordinates.
\item There are no $x\in Q(F)$ and $I\subseteq[k]$ with $\eta k\le|I|<(\eta+\delta)k$ such that every $i\in I$ is the first coordinate of some pair in $\cL_x(I)$.
\item There are no $x\in Q(F)$ and $I\subseteq[k]$ with $|I|<\eta k$ and $|\cL_x(I)|\ge k$.
\end{enumerate}

\begin{lemma}\label[lemma]{lem:rows}
Fix $x\in\Perm$ and $I\subseteq[k]$ with $|I|=d\ge1$, and condition on $x\in Q(F)$ under $\Pr_0$.
\begin{enumerate}[(a)]
\item For distinct pairs $(i_1,b_1),\dots,(i_m,b_m)$ with $i_s\in I$ and $b_s\in x(I)\setminus\{x_{i_s}\}$, the probability that all of them lie in $\cL_x(I)$ is at most $((d-1)/(k-1))^{hm}$.
\item The probability that every $i\in I$ is the first coordinate of some pair in $\cL_x(I)$ is at most $\bigl[(d-1)((d-1)/(k-1))^h\bigr]^d$.
\end{enumerate}
\end{lemma}

\begin{proof}
Under this conditioning the families are independent, and each is uniform among the families that accept $x$.
Condition further on $v^j=\pi^j(x)$ for every $j$, and apply \cref{lem:residual} to each family.
For a pair $(i,b)$ with $b\ne x_i$, the value $\pi^j_i(b)$ is uniform on $[k]\setminus\{v^j_i\}$, and $(i,b)\in\cL_x(I)$ if and only if $\pi^j_i(b)\in v^j(I)\setminus\{v^j_i\}$ for every $j$, a set of $d-1$ values.
If $n_i$ of the given pairs have first coordinate $i$, their second coordinates are distinct, so their images under $\pi^j_i$ form a uniformly random $n_i$-tuple of distinct elements of $[k]\setminus\{v^j_i\}$.
All of them lie in the $(d-1)$-element set with probability $\frac{(d-1)(d-2)\cdots(d-n_i)}{(k-1)(k-2)\cdots(k-n_i)}\le\bigl(\frac{d-1}{k-1}\bigr)^{n_i}$.
Different coordinates and different families are independent, which proves~(a).
For~(b), a union bound over the $d-1$ choices of $b$ bounds the probability for each $i\in I$ by $(d-1)((d-1)/(k-1))^h$, and different coordinates are independent.
Both bounds hold for every value of the $v^j$, hence also after averaging over them.
\end{proof}

\begin{lemma}\label[lemma]{lem:good}
For large $k$, $\Pr_0[F\text{ is not good}]\le e^{-k\ln k/600}$.
\end{lemma}

\begin{proof}
For (G1), \cref{lem:first,lem:second} bound the expected number of pairs of configurations of $Q(F)$ that differ in $d$ coordinates, for $\delta k\le d\le\beta k$, by $s_d\,\E_0Z\le\exp\bigl(k\ln k\,(\varepsilon+g(d/k))+4k\bigr)$.

For (G2) and (G3), a union bound over $x\in\Perm$ and $I\subseteq[k]$ bounds the probability that some pair $(x,I)$ violates them by $\sum_{x,I}\Pr_0[x\in Q(F)]\,P(x,I)\le\E_0Z\cdot2^k\max_{x,I}P(x,I)$, where $P(x,I)$ is the probability, conditioned on $x\in Q(F)$, that $(x,I)$ violates the property.
For (G2), let $\eta k\le d=|I|<(\eta+\delta)k$ and $u=d/k$.
Since $h>\alpha\ln k-1$, $(d-1)/(k-1)\le u$, and $(k-1)/(d-1)\le2/\eta$ for large $k$, \cref{lem:rows}(b) gives
\[
\ln P(x,I)\le d\ln k+(\alpha\ln k-1)\,d\ln\frac{d-1}{k-1}\le k\ln k\,(u+\alpha u\ln u)+k\ln\frac2\eta.
\]
For (G3), let $d=|I|<\eta k$.
Then $(d-1)/(k-1)<\eta$, and there are at most $k^2$ candidate pairs, so by \cref{lem:rows}(a), $P(x,I)\le\binom{k^2}{k}\eta^{hk}\le(ek)^k\eta^{hk}$, and $\ln P(x,I)\le k\ln k\,(1+\alpha\ln\eta)+k(1+\ln(1/\eta))$.

With \cref{lem:first}, the logarithms of the three probabilities are therefore at most $k\ln k$ times
\[
\max_{\delta\le u\le\beta}\bigl(\varepsilon+u(\varepsilon+\alpha\ln u)\bigr),\qquad\max_{\eta\le u\le\eta+\delta}\bigl(\varepsilon+u(1+\alpha\ln u)\bigr),\qquad\varepsilon+1+\alpha\ln\eta,
\]
plus $O(k)$.
Both functions of $u$ are convex, so each maximum is attained at an endpoint; the values are $-0.0118\ldots$ and $-0.0041\ldots$ at $\delta$ and $\beta$, $-0.0018\ldots$ and $-0.0089\ldots$ at $\eta$ and $\eta+\delta$, and the last expression is $-1.17\ldots$.
So the sum of the three probabilities is at most $e^{-k\ln k/600}$ for large $k$.
\end{proof}

\begin{lemma}\label[lemma]{lem:fixed}
For large $k$, there is a good system $F$ with $Z(F)\ge\E_0Z/2$.
\end{lemma}

\begin{proof}
Since $\E_0\bigl[Z\,\one[Z<\E_0Z/2]\bigr]\le\E_0Z/2$, the Cauchy--Schwarz inequality and \cref{lem:second} give
\[
\tfrac12\E_0Z\le\E_0\bigl[Z\,\one[Z\ge\E_0Z/2]\bigr]\le\bigl(\E_0[Z^2]\,\Pr_0[Z\ge\E_0Z/2]\bigr)^{1/2}\le\bigl(e^{5k}(\E_0Z)^2\,\Pr_0[Z\ge\E_0Z/2]\bigr)^{1/2},
\]
so $\Pr_0[Z\ge\E_0Z/2]\ge e^{-5k}/4$.
By \cref{lem:good}, this exceeds $\Pr_0[F\text{ is not good}]$ for large $k$.
\end{proof}

\begin{lemma}\label[lemma]{lem:structure}
Let $F$ be good.
Configurations in distinct components of $Q(F)$ differ in more than $\beta k$ coordinates, and every component $C$ of $Q(F)$ varies in fewer than $\eta k$ coordinates and satisfies $\sum_{i=1}^k(|A_i(C)|-1)<k$; in particular $|C|<2^k$.
\end{lemma}

\begin{proof}
Configurations in distinct components are not joined by an edge, so they differ in at least $\delta k$ coordinates, and hence, by (G1), in more than $\beta k$.

Fix $x\in C$, and let $I_C$ be the set of coordinates in which $C$ varies, that is, in which some $y\in C$ differs from $x$.
Since $C$ is connected, it can be ordered as $x=y^0,y^1,\dots,y^n$ so that every $y^s$ with $s\ge1$ differs in fewer than $\delta k$ coordinates from some $y^{s'}$ with $s'<s$.
Let $I_s$ be the set of coordinates in which at least one of $y^0,\dots,y^s$ differs from $x$; so $y^0,\dots,y^s$ agree with $x$ outside $I_s$.
Since $y^{s'}$ agrees with $x$ outside $I_{s-1}$, every coordinate in $I_s\setminus I_{s-1}$ is one in which $y^s$ differs from $y^{s'}$, so $|I_s\setminus I_{s-1}|<\delta k$.
Suppose that $|I_n|\ge\eta k$, and let $s$ be the least index with $|I_s|\ge\eta k$; then $\eta k\le|I_s|<(\eta+\delta)k$.
Every $i\in I_s$ has some $s''\le s$ with $y^{s''}_i\ne x_i$, and $y^{s''}$ agrees with $x$ outside $I_s$, so $(i,y^{s''}_i)\in\cL_x(I_s)$ by \cref{lem:alt}.
This contradicts (G2), so $|I_C|=|I_n|<\eta k$.

Every $y\in C$ agrees with $x$ outside $I_C$.
By \cref{lem:alt}, every pair $(i,b)$ with $b\in A_i(C)\setminus\{x_i\}$ lies in $\cL_x(I_C)$, so (G3) gives $\sum_i(|A_i(C)|-1)\le|\cL_x(I_C)|<k$.
Finally, $|C|\le\prod_i|A_i(C)|\le\prod_i2^{|A_i(C)|-1}<2^k$.
\end{proof}

\subsection{Requests that forbid one position}\label{sec:elim}

A request $r$ \emph{eliminates} a set $C$ of configurations if no configuration in $C$ serves $r$.
For $i,a\in[k]$, let $\langle i,a\rangle$ be the request whose $i$-th coordinate is $a+1$ (or $1$ if $a=k$) and whose other coordinates are all $a$.

\begin{lemma}\label[lemma]{lem:forbid}
A permutation $x$ serves $\langle i,a\rangle$ if and only if $x_i\ne a$.
\end{lemma}

\begin{proof}
If $x_i=a$, then $x_\ell\ne a$ for every $\ell\ne i$, and $x_i$ differs from the $i$-th coordinate of $\langle i,a\rangle$, so $x$ fails $\langle i,a\rangle$.
If $x_i\ne a$, then $x_\ell=a$ for some $\ell\ne i$, since $x$ is a permutation, so $x$ serves $\langle i,a\rangle$.
\end{proof}

\begin{lemma}\label[lemma]{lem:elim}
Let $F$ be good, let $C$ be a component of $Q(F)$, and let $\sigma\in Q(F)\setminus C$.
Let $(i,a)$ be uniformly random among the $k(k-1)$ pairs in $[k]^2$ with $\sigma_i\ne a$.
Then $\langle i,a\rangle$ eliminates $C$ with probability at least $(\beta-\eta)/(k-1)$.
\end{lemma}

\begin{proof}
Fix $x\in C$.
By \cref{lem:structure}, $\sigma$ differs from $x$ in more than $\beta k$ coordinates, and $C$ varies in fewer than $\eta k$ coordinates.
So more than $(\beta-\eta)k$ coordinates $i$ satisfy $\sigma_i\ne x_i$ and $y_i=x_i$ for every $y\in C$.
For each of them, $(i,x_i)$ is one of the $k(k-1)$ pairs, and $\langle i,x_i\rangle$ eliminates $C$ by \cref{lem:forbid}.
\end{proof}

\subsection{The posterior of the hidden permutation}\label{sec:posterior}

From now on, fix a good system $F$ with $Z=Z(F)\ge\E_0Z/2$, which exists for large $k$ by \cref{lem:fixed}, and let $T=\lfloor(k-1)(\varepsilon k\ln k-4k)\rfloor$.
Let $\sigma$ be uniformly random on $Q(F)$, and, given $\sigma$, let $(i_1,a_1),\dots,(i_T,a_T)$ be independent, each uniformly random among the pairs $(i,a)\in[k]^2$ with $\sigma_i\ne a$.
Let $r_t=\langle i_t,a_t\rangle$, and for $0\le t\le T$, let $Q_t$ be the set of configurations in $Q(F)$ that serve $r_1,\dots,r_t$.

\begin{lemma}\label[lemma]{lem:posterior}
For $0\le t\le T$, given $(i_1,a_1),\dots,(i_t,a_t)$, the permutation $\sigma$ is uniform on $Q_t$, and $2^k\,\E[1/|Q_t|]\le e^{-2k}$ for large $k$.
\end{lemma}

\begin{proof}
By \cref{lem:forbid}, a permutation $x\in Q(F)$ lies in $Q_t$ if and only if $x_{i_s}\ne a_s$ for every $s\le t$.
Hence $\Pr[\sigma=x,(i_1,a_1),\dots,(i_t,a_t)]=Z^{-1}(k(k-1))^{-t}\,\one[x\in Q_t]$, which, for fixed pairs, is proportional to $\one[x\in Q_t]$.
Summing over $x$ shows that the pairs have probability $|Q_t|/(Z(k(k-1))^t)$, so $\E[1/|Q_t|]$ is the number of sequences of $t$ pairs with $Q_t\ne\emptyset$ divided by $Z(k(k-1))^t$, which is at most $k^{2t}/(Z(k(k-1))^t)=1/(Z(1-1/k)^t)$.
Since $\ln(1-1/k)\ge-1/(k-1)$ and $t\le T$, we have $(1-1/k)^{-t}\le e^{\varepsilon k\ln k-4k}$, and $Z\ge\E_0Z/2\ge e^{\varepsilon k\ln k-k}/2$ by \cref{lem:first}.
So $2^k\,\E[1/|Q_t|]\le2^{k+1}e^{-3k}\le e^{-2k}$.
\end{proof}

\subsection{The adversary}\label{sec:adversary}

Fix a randomized online algorithm $\ALG$ for $\cX=[k]^k$ with starting configuration $a$.
If $\ALG$ is not $\gamma$-competitive for any $\gamma$, there is nothing to prove.
Otherwise fix $\gamma$ and $c_\ALG$ such that $\E[\ALG(\rho)]\le\gamma\OPT(\rho)+c_\ALG$ for every finite request sequence $\rho$.
For a finite request sequence $\rho$, let $\mu_\rho$ be the distribution of the configuration of $\ALG$ after it serves $\rho$; it is determined by $\ALG$ and $\rho$.

\begin{lemma}\label[lemma]{lem:confine}
Let $\rho_0$ be a finite request sequence, let $\mathcal R$ be a finite set of requests that occur in $\rho_0$ such that the set $Q$ of configurations serving every request in $\mathcal R$ is nonempty, and let $\tau>0$.
There is a finite sequence $\rho_1$ of requests from $\mathcal R$, determined by $\ALG$, $\rho_0$, $\mathcal R$ and $\tau$, such that $\mu_{\rho_0\rho_1}(\cX\setminus Q)\le\tau$.
\end{lemma}

\begin{proof}
We build $\rho_1$ one request at a time.
While the current sequence $\rho'$ satisfies $\mu_{\rho'}(\cX\setminus Q)>\tau$, let $u$ be the first configuration outside $Q$, in a fixed order of $\cX$, with $\mu_{\rho'}(u)>\tau/|\cX|$, and append the first request of $\mathcal R$, in a fixed order, that $u$ fails, which exists because $u\notin Q$.
With probability $\mu_{\rho'}(u)$ the algorithm is at $u$ and must move, so each appended request raises the expected cost of $\ALG$ by more than $\tau/|\cX|$.
On the other hand, an offline solution can serve $\rho_0$ optimally, move to a fixed configuration of $Q$ at the first appended request, at cost at most $k$, and stay there, since every appended request lies in $\mathcal R$.
So every sequence $\rho'$ built this way has $\OPT(\rho')\le\OPT(\rho_0)+k$.
If the construction never stopped, $\E[\ALG(\rho')]$ would eventually exceed $\gamma(\OPT(\rho_0)+k)+c_\ALG$, a contradiction.
\end{proof}

\paragraph{One phase.}
Let $\tau=k^{-3}$.
A \emph{phase} draws $\sigma$ and $(i_1,a_1),\dots,(i_T,a_T)$ as in \cref{sec:posterior}, independently of all earlier phases, and issues the following requests.
\begin{enumerate}[(1)]
\item The setup requests of $F$, followed by the confinement sequence of \cref{lem:confine} for the set $\mathcal R_0$ of setup requests.
\item For $t=1,\dots,T$: the request $r_t$, followed by the confinement sequence for $\mathcal R_t=\mathcal R_0\cup\{r_1,\dots,r_t\}$.
\end{enumerate}
By \cref{lem:enforce}, the set of configurations serving every request of $\mathcal R_t$ is $Q_t$, which contains $\sigma$.
Each confinement sequence is computed from the distributions $\mu_\rho$ of $\ALG$, never from its random choices.
Hence the requests of any number of phases form a random sequence that is independent of the random choices of $\ALG$, as required for an oblivious adversary, and it takes finitely many values.

\begin{lemma}\label[lemma]{lem:phase}
For large $k$, the expected cost of $\ALG$ during each phase is at least $\frac1{20}k^2\ln k$.
\end{lemma}

\begin{proof}
Fix a phase and $t\in\{0,\dots,T-1\}$.
Let $Y_t$ be the configuration of $\ALG$ at the end of the confinement sequence that follows $r_t$ (for $t=0$, the one that follows the setup requests).
When $Y_t\in Q(F)$, let $C(Y_t)$ be the component of $Q(F)$ that contains it.
Let $\cH_t$ consist of the draws of all earlier phases together with $(i_1,a_1),\dots,(i_t,a_t)$ of the current phase.
The requests issued up to time $Y_t$ are a function of $\cH_t$, and the random choices of $\ALG$ are independent of $(\sigma,\cH_t)$.
So, given $\cH_t$, the configuration $Y_t$ is independent of $\sigma$, and since $\sigma$ is independent of the earlier phases, \cref{lem:posterior} shows that $\sigma$ is uniform on $Q_t$ given $(\cH_t,Y_t)$.

Let $E$ be the event that $Y_t\in Q_t$ and $\sigma\notin C(Y_t)$.
If $Y_t\in Q_t$, then by \cref{lem:structure} the conditional probability that $\sigma\in C(Y_t)$, given $(\cH_t,Y_t)$, is at most $2^k/|Q_t|$.
With \cref{lem:posterior,lem:confine}, this gives $\Pr[\overline E]\le2^k\,\E[1/|Q_t|]+\tau\le e^{-2k}+\tau$.

Given $\sigma$, the pair $(i_{t+1},a_{t+1})$ is independent of $(\cH_t,Y_t)$ and uniform among the pairs $(i,a)$ with $\sigma_i\ne a$.
So by \cref{lem:elim}, on $E$ the request $r_{t+1}$ eliminates $C(Y_t)$ with conditional probability at least $(\beta-\eta)/(k-1)$.
If it does and $Y_{t+1}\in Q_{t+1}$, then $Y_{t+1}\in Q(F)\setminus C(Y_t)$, so by \cref{lem:structure} $\ALG$ moves a distance of more than $\beta k$ between the two times.
Since $\Pr[Y_{t+1}\notin Q_{t+1}]\le\tau$, the expected cost of $\ALG$ between the two times is at least $\beta k\bigl((\beta-\eta)(1-e^{-2k}-\tau)/(k-1)-\tau\bigr)\ge\beta(\beta-\eta)-1/k$ for large $k$.
Summing over $t=0,\dots,T-1$ and using $T\ge\varepsilon k^2\ln k-5k^2$ and $\varepsilon\beta(\beta-\eta)=0.05313>\frac1{20}$ proves the lemma for large $k$.
\end{proof}

\begin{proof}[Proof of \cref{thm:lower}]
Suppose that $\ALG$ is $\gamma$-competitive with additive constant $c_\ALG$, and let the adversary run $L$ phases.
In each phase, the offline solution that moves to that phase's $\sigma$ at the first request of the phase and stays there serves every request of the phase: $\sigma$ serves the setup requests and, by \cref{lem:forbid}, every $r_t$, and every confinement request repeats one of them.
So $\OPT\le Lk$ for every outcome.
By \cref{lem:phase}, $\ALG$ has expected cost at least $\frac1{20}Lk^2\ln k$.
Since the input is independent of the random choices of $\ALG$ and takes finitely many values, averaging the competitive guarantee over the input gives $\frac1{20}Lk^2\ln k\le\gamma Lk+c_\ALG$.
Letting $L\to\infty$ yields $\gamma\ge\frac1{20}k\ln k$.
\end{proof}

\section{A \texorpdfstring{$20^k$}{20\^{}k}-competitive algorithm for weighted uniform metrics}\label{sec:weighted}

In this section the weights are arbitrary, and we prove \cref{thm:weighted}.
We use the phases of Bijoy, Mondal, and Chiplunkar~\cite{BMC26}, stated for generalized requests, and replace their rule for following virtual algorithms by the aggregation of \cref{lem:aggregate}.
Bijoy et al.\ sketch how their phases and their offline lower bound carry over to generalized requests~\cite[Appendix~C]{BMC26}; we give complete proofs in this setting.
As in their sketch, demand is counted separately in every coordinate (\cref{lem:basic}(d)).
\Cref{sec:phases} defines the phases, \cref{sec:aggregate} proves the aggregation lemma, \cref{sec:strategies} bounds the cost of our strategy for a phase, \cref{sec:offline} shows that every offline solution pays a $2^{-k}$ fraction of the largest weight per completed phase of level $k$, and \cref{sec:weighted-proof} combines these bounds.

\paragraph{Normalized weights.}
Until \cref{sec:weighted-proof}, we assume $k\ge2$ and that the weights, now written $v_1,\dots,v_k$, are positive and satisfy
\[
m_\ell:=v_{\ell+1}/v_\ell\in\{2,3,\dots\}\qquad(1\le\ell<k).
\]
Following~\cite{BMC26}, let $d_\ell=2^{5^{\ell-1}-1}$, so $d_1=1$ and $d_2=16$; we also assume that $|M_\ell|\ge d_\ell-1$ for every $\ell\ge2$.
Movement costs are measured by $d_v$, and $\sum_{i\le\ell}v_i\le2v_\ell$ for every $\ell$.

\subsection{Phases}\label{sec:phases}

Let $M=M_1\cup\dots\cup M_k$.
For $1\le\ell\le k$, let $\cU_\ell$ be the family of sets $U\subseteq M_{\ell+1}\cup\dots\cup M_k$ that contain at most one point of each $M_j$.
A set $U$ \emph{satisfies} a request $r$ if $r[j]\in U$ for some $j$, and a configuration $x$ \emph{covers} $U$ if $x[j]=u$ for every $u\in U\cap M_j$.
A configuration that covers $U$ serves every request that $U$ satisfies.
A \emph{demand vector} is a finitely supported function $V:M\to[0,\infty)$; $V|_j$ is its restriction to $M_j$, $\|V|_j\|=\sum_{p\in M_j}V[p]$, and $\one_p$ is the indicator vector of a point $p$.

A \emph{parser} reads a stream of requests.
At each request it either accepts the request or stops; when it stops it does not consume the current request, which is then read by whatever runs next.
In our strategy for an $(\ell,U)$-phase (\cref{lem:strategy}), only servers $1,\dots,\ell$ move and the configuration covers $U$, so the requests that $U$ satisfies cost nothing.
The witness of a $(1,U)$-phase, defined next, is its first request that $U$ does not satisfy; the strategy then places server~$1$ at the first coordinate of the witness, and the phase stops at the first later request that neither $U$ nor this position of server~$1$ satisfies.

\begin{definition}[phases and multiphases, after~\cite{BMC26}]\label[definition]{def:phase}
Let $1\le\ell\le k$ and $U\in\cU_\ell$.
\begin{itemize}
\item The \emph{$(1,U)$-phase} accepts every request that $U$ satisfies, until it reads a request $r^*$ that $U$ does not satisfy, its \emph{witness}.
It accepts $r^*$, then accepts every request that $U\cup\{r^*[1]\}$ satisfies, and stops at the first request that $U\cup\{r^*[1]\}$ does not satisfy.
Its demand vector is $\sum_{j=1}^k\one_{r^*[j]}$.
\item For $\ell<k$, the \emph{$(\ell,U)$-multiphase} runs $m_\ell$ copies of the $(\ell,U)$-phase one after another, each starting with the request at which the previous one stopped; it stops when the last one stops, and its demand vector is the sum of theirs.
\item For $\ell\ge2$, the $(\ell,U)$-phase first runs an $(\ell-1,U)$-multiphase, its \emph{explore part}, with demand vector $V^0$.
When the explore part stops, the \emph{critical set} $S$ consists of the $d_\ell-1$ points $p\in M_\ell$ with the largest values $V^0[p]$, with ties broken by a fixed well-order of $M_\ell$.
From the request at which the explore part stopped on, the phase runs, for every $p\in S$, an $(\ell-1,U\cup\{p\})$-multiphase, its \emph{exploit part} for $p$; all of them read the same requests, each until it stops.
The phase accepts a request as long as at least one exploit part accepts it, and stops when all of them have stopped.
Its demand vector is $V^0+\sum_{p\in S}V^p$, where $V^p$ is the demand vector of the exploit part for $p$.
\end{itemize}
A phase or multiphase is \emph{completed} if it has stopped; demand vectors are used only for completed ones.
\end{definition}

Note that $U\cup\{p\}\in\cU_{\ell-1}$ for $p\in M_\ell$.
The request sequence is parsed into consecutive $(k,\emptyset)$-phases, each starting with the request at which the previous one stopped.
For example, if the first coordinates of the requests are $a,a,b,b,c,\dots$, the first $(1,\emptyset)$-phase has the first request as its witness, accepts the second, and stops at the third, which is the witness of the next $(1,\emptyset)$-phase; that phase accepts the fourth and stops at the fifth.

\begin{lemma}\label[lemma]{lem:basic}
\begin{enumerate}[(a)]
\item Every phase and multiphase accepts the first request it reads; in particular, every completed one has accepted at least one request.
\item Which requests a phase or multiphase accepts, when it stops, its critical sets, and its demand vector once it is completed are determined by $\ell$, $U$ and the requests it reads.
\item If $r_1^*$ and $r_2^*$ are the witnesses of two consecutive phases of a completed $(1,U)$-multiphase, then $r_1^*[1]\ne r_2^*[1]$.
\item For every $j\in[k]$, the demand vector $V$ of a completed $(\ell,U)$-phase satisfies $\|V|_j\|=\kappa_\ell$, where $\kappa_\ell=(v_\ell/v_1)\prod_{i=1}^\ell d_i$, and that of a completed $(\ell,U)$-multiphase satisfies $\|V|_j\|=m_\ell\kappa_\ell$.
Moreover $V[u]=0$ for every $u\in U$.
\end{enumerate}
\end{lemma}

\begin{proof}
(a) A $(1,U)$-phase accepts its first request: either $U$ satisfies it, or it is the witness.
A multiphase starts with a phase, and a phase of level $\ell\ge2$ starts with its explore part, so (a) follows by induction on $\ell$.
(b) holds because all rules are deterministic.

(c) The first phase stops at a request that $U\cup\{r_1^*[1]\}$ does not satisfy.
This request is therefore not satisfied by $U$, and its first coordinate differs from $r_1^*[1]$.
It is the first request read by the second phase, which does not satisfy $U$, so it is the witness $r_2^*$.

(d) By induction on $\ell$, as in~\cite[Lemma~46]{BMC26}.
For $\ell=1$, the demand vector $\sum_j\one_{r^*[j]}$ has one unit in each $M_j$, and $\kappa_1=d_1=1$; since $U$ does not satisfy $r^*$, no point of $U$ receives demand.
A multiphase sums $m_\ell$ phases.
For $\ell\ge2$, the explore part and the $d_\ell-1$ exploit parts of a completed phase are completed $(\ell-1)$-multiphases, each with $\|V|_j\|=m_{\ell-1}\kappa_{\ell-1}$, so the phase has $\|V|_j\|=d_\ell m_{\ell-1}\kappa_{\ell-1}=\kappa_\ell$.
Each of these parts belongs to a set that contains $U$, so by induction its demand vector vanishes on $U$.
\end{proof}

\subsection{Aggregating parallel runs}\label{sec:aggregate}

The next lemma is a version of the experts algorithms with switching costs of Blum and Burch~\cite{BB00}, in which experts drop out over time.

\begin{lemma}\label[lemma]{lem:aggregate}
Let $(\cX',d')$ be a metric space of diameter at most $\Delta>0$, and consider $n$ experts over $t^*$ steps.
Initially every expert $i$ has a position $x_i(0)\in\cX'$.
At step $t=1,\dots,t^*$, a nonempty set $\cA_t\subseteq\cA_{t-1}$ of active experts is revealed, where $\cA_0=[n]$, together with a position $x_i(t)$ and a cost $\lambda_i(t)\ge d'(x_i(t-1),x_i(t))$ for every $i\in\cA_t$.
There is a randomized master that starts at an arbitrary point of $\cX'$, is at $x_{I_t}(t)$ for some $I_t\in\cA_t$ after step $t$, and satisfies, for every fixed choice of the active sets, positions and costs and every expert $j\in\cA_{t^*}$,
\[
\E[\text{movement of the master}]\le\Delta+2\Delta\ln n+2\sum_{t=1}^{t^*}\lambda_j(t),
\]
where the expectation is over the random choices of the master only.
\end{lemma}

\begin{proof}
Let $\omega_i(0)=1$, and for $t\ge1$ let $\omega_i(t)=\omega_i(t-1)e^{-\lambda_i(t)/\Delta}$ for $i\in\cA_t$ and $\omega_i(t)=0$ otherwise.
Let $\nu_t=\omega(t)/\|\omega(t)\|$.
The master draws $I_0$ from $\nu_0$ and moves to $x_{I_0}(0)$, at cost at most $\Delta$.
At step $t$, given $I_{t-1}=i$, it keeps $I_t=i$ with probability $\min\{1,\nu_t(i)/\nu_{t-1}(i)\}$, and otherwise draws $I_t$ with probabilities proportional to $(\nu_t(i')-\nu_{t-1}(i'))_+$; then it moves to $x_{I_t}(t)$.
This maximal coupling gives $I_t\sim\nu_t$ and $\Pr[I_t\ne I_{t-1}]=\mathrm{TV}(\nu_{t-1},\nu_t)$.

Fix $t$, and let $\zeta_t=\sum_{i\in\cA_t}\nu_{t-1}(i)e^{-\lambda_i(t)/\Delta}=\|\omega(t)\|/\|\omega(t-1)\|\in(0,1]$ and $\psi_t=-\ln\zeta_t$.
Since $\nu_t(i)\ge\nu_{t-1}(i)e^{-\lambda_i(t)/\Delta}\one[i\in\cA_t]$, we get $\mathrm{TV}(\nu_{t-1},\nu_t)\le\sum_i\nu_{t-1}(i)\bigl(1-e^{-\lambda_i(t)/\Delta}\one[i\in\cA_t]\bigr)=1-\zeta_t\le\psi_t$.
On the support of $\nu_t$, $\ln(\nu_t(i)/\nu_{t-1}(i))=-\lambda_i(t)/\Delta+\psi_t$, and the relative entropy of $\nu_t$ with respect to $\nu_{t-1}$ is nonnegative, so $\sum_i\nu_t(i)\lambda_i(t)\le\Delta\psi_t$.
The master moves at most $\Delta$ when it switches experts and $\lambda_{I_t}(t)$ otherwise, so its expected movement at step $t$ is at most $\Delta\,\mathrm{TV}(\nu_{t-1},\nu_t)+\sum_i\nu_t(i)\lambda_i(t)\le2\Delta\psi_t$.
Finally, $\sum_t\psi_t=\ln\|\omega(0)\|-\ln\|\omega(t^*)\|\le\ln n+\sum_t\lambda_j(t)/\Delta$, because $\omega_j(t^*)=e^{-\sum_t\lambda_j(t)/\Delta}$.
\end{proof}

\subsection{Phase strategies}\label{sec:strategies}

The phases are determined by the requests alone (\cref{lem:basic}(b)); our strategy adds two independent sources of randomness, the random choices of the simulated copies and those of the master of \cref{lem:aggregate}, which chooses the copy that the real servers follow.
Let $c_1=1$ and $c_\ell=3c_{\ell-1}+4\ln(d_\ell-1)+2$ for $\ell\ge2$.

\begin{lemma}\label[lemma]{lem:strategy}
For every $1\le\ell\le k$ and $U\in\cU_\ell$, there is a randomized online strategy for the $(\ell,U)$-phase with the following properties, for every fixed request stream and every starting configuration that covers $U$.
\begin{enumerate}[(a)]
\item It moves only servers $1,\dots,\ell$.
\item It serves every request that the phase accepts, and it stops when the phase stops.
\item Its expected cost is at most $c_\ell v_\ell$, also when the stream ends before the phase stops.
\end{enumerate}
For $\ell<k$, running it for the $m_\ell$ phases of an $(\ell,U)$-multiphase gives expected cost at most $c_\ell v_{\ell+1}$.
\end{lemma}

\begin{proof}
For $\ell=1$, servers $2,\dots,k$ stay where they are and serve every request that $U$ satisfies; when the witness $r^*$ arrives, server $1$ moves to $r^*[1]$, at cost at most $v_1$, and the configuration then covers $U\cup\{r^*[1]\}$, so it serves every later request of the phase.

For a multiphase, servers $\ell+1,\dots,k$ never move, so each of the $m_\ell$ phases starts from a configuration that covers $U$; each phase uses fresh randomness and costs at most $c_\ell v_\ell$ in expectation given the past, so the multiphase costs at most $m_\ell c_\ell v_\ell=c_\ell v_{\ell+1}$.

Let $\ell\ge2$.
The strategy first runs the strategy for the $(\ell-1,U)$-multiphase on the real servers while the explore part runs, at expected cost at most $c_{\ell-1}v_\ell$.
When the explore part stops, it computes $S$ and starts, for every $p\in S$, a simulated copy of the strategy for the $(\ell-1,U\cup\{p\})$-multiphase with fresh independent randomness, from the current configuration with server $\ell$ moved to $p$; this configuration covers $U\cup\{p\}$.
Every copy reads every request until its multiphase stops.
The real servers follow the master of \cref{lem:aggregate}, whose experts are the copies, with their configurations as positions, their movement costs as costs, and the copies that have not stopped as active experts.
All these configurations agree on servers $\ell+1,\dots,k$, so their distances are at most $\sum_{i\le\ell}v_i\le2v_\ell=:\Delta$.
After each accepted request the real servers are at the configuration of an active copy, which serves the request, and when all copies have stopped, the phase stops.
So (a) and (b) hold.

For (c), condition on everything up to the end of the explore part; this fixes $S$ and the starting configurations of the copies.
By \cref{lem:basic}(b), when each copy stops is determined by the request stream, independently of all random choices.
Let $t^*$ be the number of requests that the exploit part accepts before it stops or the stream ends.
If $t^*=0$, the exploit part costs nothing.
Otherwise let $j$ be the first copy, in a fixed order, that is still active after these $t^*$ requests; $j$ does not depend on any random choice.
For every fixed outcome of the randomness of the copies, \cref{lem:aggregate} bounds the expected movement of the master by $\Delta+2\Delta\ln(d_\ell-1)+2\lambda_j$, where $\lambda_j$ is the movement of copy $j$ during these $t^*$ requests.
Averaging over the randomness of copy $j$, and applying the multiphase bound at level $\ell-1$ to the first $t^*$ requests, gives $\E[\lambda_j]\le c_{\ell-1}v_\ell$.
So the exploit part costs at most $(2+4\ln(d_\ell-1)+2c_{\ell-1})v_\ell$ in expectation, and the phase at most $(3c_{\ell-1}+4\ln(d_\ell-1)+2)v_\ell=c_\ell v_\ell$.
\end{proof}

Since $\ln(d_\ell-1)<(5^{\ell-1}-1)\ln2\le5^{\ell-1}-1$, induction gives $c_\ell\le3\cdot10\cdot5^{\ell-2}+4(5^{\ell-1}-1)+2\le10\cdot5^{\ell-1}$ for $\ell\ge2$, and $c_1=1$; so
\begin{equation}\label{eq:ck}
c_k\le10\cdot5^{k-1}.
\end{equation}

\subsection{The offline lower bound}\label{sec:offline}

Let $\rho$ be the sequence of requests accepted by a completed phase or multiphase.
A \emph{solution} for $\rho$ is a sequence $\xi=(\xi_0,\xi_1,\dots,\xi_n)$ of configurations, where $n=|\rho|$ and $\xi_t$ serves the $t$-th request of $\rho$ for $t\ge1$; the configuration $\xi_0$ is arbitrary.
Its cost is $\sum_{t=1}^nd_v(\xi_{t-1},\xi_t)$.
The \emph{restriction} of $\xi$ to a part, such as a phase of a multiphase, consists of the configuration just before the first request of the part, followed by the configurations at its requests.
The phases of a multiphase accept consecutive disjoint intervals of requests, so the costs of the restrictions of $\xi$ to them add up to at most the cost of $\xi$.

\begin{definition}[hard multiphase and contamination, after~\cite{BMC26}]\label[definition]{def:contamination}
Let $P$ be a completed $(\ell,U)$-phase and $\xi$ a solution for it.
\begin{itemize}
\item If $\ell\ge2$, the \emph{hard multiphase} $\cQ(P,\xi)$ is the explore part of $P$ if $\xi_0[\ell]\notin S$, and the exploit part for $\xi_0[\ell]$ if $\xi_0[\ell]\in S$.
\item For $j>\ell$, the solution $\xi$ is \emph{$j$-contaminated} for $P$ if either $\ell=1$ and $\xi_0[j]=r^*[j]$, where $r^*$ is the witness of $P$; or $\ell\ge2$ and, for at least $2^{-(j-\ell+3)}m_{\ell-1}$ of the $m_{\ell-1}$ phases $P'$ of $\cQ(P,\xi)$, the restriction of $\xi$ to $P'$ is $j$-contaminated for~$P'$.
\item The solution $\xi$ is \emph{contaminated} for $P$ if it is $j$-contaminated for some $j>\ell$.
\end{itemize}
\end{definition}

Informally, the hard multiphase is the part of $P$ in which server~$\ell$, kept at $\xi_0[\ell]$, serves few witnesses (\cref{lem:hard}), and a $j$-contaminated solution in which server~$j$ stays still keeps it at a point that many witnesses of $P$ request (\cref{lem:demand}).
\Cref{lem:phasecost} shows that a solution that is not contaminated for a phase of level $\ell\ge2$ pays at least $v_\ell/2^\ell$; for $\ell=k$, no solution is contaminated.

Define $\Gamma_{1,j}=1$ and $\Gamma_{\ell,j}=2^{j-\ell+3}d_{\ell-1}\Gamma_{\ell-1,j}$ for $2\le\ell\le j$.

\begin{lemma}[after~{\cite[Claim~33]{BMC26}}]\label[lemma]{lem:gamma}
$d_\ell\ge2\Gamma_{\ell,\ell}$ for every $\ell\ge2$.
\end{lemma}

\begin{proof}
Unrolling the definition,
\[
\log_2\Gamma_{\ell,\ell}=\sum_{i=2}^\ell(\ell-i+3)+\sum_{i=1}^{\ell-1}(5^{i-1}-1)=\frac{\ell^2+3\ell-4}2+\frac{5^{\ell-1}-1}4-(\ell-1)=\frac{\ell^2+\ell-2}2+\frac{5^{\ell-1}-1}4.
\]
So the claim is equivalent to $\frac34(5^{\ell-1}-1)\ge\frac{\ell^2+\ell}2$, that is, $5^{\ell-1}-1\ge\frac23\ell(\ell+1)$.
This holds with equality for $\ell=2$, and if it holds for $\ell$, then $5^\ell-1=5(5^{\ell-1}-1)+4\ge\frac{10}3\ell(\ell+1)\ge\frac23(\ell+1)(\ell+2)$.
\end{proof}

\begin{lemma}[after~{\cite[Lemmas~34 and~48]{BMC26}}]\label[lemma]{lem:demand}
Let $P$ be a completed $(\ell,U)$-phase with demand vector $V$, let $j>\ell$, and let $\xi$ be a solution for $P$ in which server $j$ stays at a point $b$.
If $\xi$ is $j$-contaminated for $P$, then $V[b]\ge\|V|_j\|/(d_\ell\Gamma_{\ell,j})$.
\end{lemma}

\begin{proof}
By induction on $\ell$.
For $\ell=1$, contamination means $b=\xi_0[j]=r^*[j]$, so $V[b]=1=\|V|_j\|$, and $d_1\Gamma_{1,j}=1$.
Let $\ell\ge2$ and $m=m_{\ell-1}$, and let $V_1,\dots,V_m$ be the demand vectors of the phases of $\cQ(P,\xi)$.
Server $j$ stays at $b$ in each restriction of $\xi$, and at least $m/2^{j-\ell+3}$ restrictions are $j$-contaminated.
By induction and \cref{lem:basic}(d), each of them has $V_i[b]\ge\kappa_{\ell-1}/(d_{\ell-1}\Gamma_{\ell-1,j})$.
The vectors $V_i$ are summands of $V$, so
\[
V[b]\ge\frac{m}{2^{j-\ell+3}}\cdot\frac{\kappa_{\ell-1}}{d_{\ell-1}\Gamma_{\ell-1,j}}=\frac{m\kappa_{\ell-1}}{\Gamma_{\ell,j}}=\frac{\kappa_\ell}{d_\ell\Gamma_{\ell,j}}=\frac{\|V|_j\|}{d_\ell\Gamma_{\ell,j}}.\qedhere
\]
\end{proof}

\begin{lemma}[after~{\cite[Lemma~18]{BMC26}}]\label[lemma]{lem:hard}
Let $P$ be a completed $(\ell,U)$-phase with $\ell\ge2$, and let $\xi$ be a solution for $P$ in which server $\ell$ stays at a point $b$.
Then the restriction of $\xi$ is $\ell$-contaminated for at most $m_{\ell-1}/16$ of the phases of $\cQ(P,\xi)$.
\end{lemma}

\begin{proof}
Let $m=m_{\ell-1}$.
If $b\in S$, then $\cQ(P,\xi)$ is the exploit part for $b$, an $(\ell-1,U\cup\{b\})$-multiphase.
By \cref{lem:basic}(d) the demand vector of each of its phases vanishes at $b$, while by \cref{lem:demand} an $\ell$-contaminated restriction would give it positive demand at $b$; so no restriction is $\ell$-contaminated.

If $b\notin S$, then $\cQ(P,\xi)$ is the explore part, with demand vector $V^0=\sum_{i=1}^mV_i$ and $\|V^0|_\ell\|=m\kappa_{\ell-1}$.
The $d_\ell-1$ points of $S$ have demand at least $V^0[b]$, so $d_\ell V^0[b]\le m\kappa_{\ell-1}$.
If $s$ restrictions are $\ell$-contaminated, \cref{lem:demand} at level $\ell-1$ gives $V^0[b]\ge s\,\kappa_{\ell-1}/(d_{\ell-1}\Gamma_{\ell-1,\ell})$.
Hence, by \cref{lem:gamma},
\[
s\le\frac{m\,d_{\ell-1}\Gamma_{\ell-1,\ell}}{d_\ell}=\frac{m\,\Gamma_{\ell,\ell}}{8d_\ell}\le\frac m{16}.\qedhere
\]
\end{proof}

\begin{lemma}[after~{\cite[Lemma~19]{BMC26}}]\label[lemma]{lem:phasecost}
Let $P$ be a completed $(\ell,U)$-phase with $\ell\ge2$, and let $\xi$ be a solution for $P$ that is not contaminated for $P$.
Then $\xi$ costs at least $v_\ell/2^\ell$.
\end{lemma}

\begin{proof}
If some server $i\ge\ell$ moves, the cost is at least $v_\ell$.
Otherwise servers $\ell,\dots,k$ stay where they are.
Let $m=m_{\ell-1}$, and call a phase $P'$ of $\cQ(P,\xi)$ \emph{good} if the restriction of $\xi$ to $P'$ is not $j$-contaminated for $P'$ for any $j\ge\ell$.
By \cref{lem:hard}, at most $m/16$ phases fail this for $j=\ell$, and since $\xi$ is not contaminated for $P$, fewer than $2^{-(j-\ell+3)}m$ fail it for each $j>\ell$.
So fewer than $m/16+m/8<m/4$ phases are not good.

For $\ell=2$, the phases of $\cQ(P,\xi)$ are $(1,U')$-phases with witnesses $r_1^*,\dots,r_m^*$.
For a good phase $i$, no server $j\ge2$ is at $r_i^*[j]$, because these servers stay at their initial positions and $\xi$ restricted to the phase is not $j$-contaminated.
So server $1$ is at $r_i^*[1]$ when $r_i^*$ arrives.
By \cref{lem:basic}(c), whenever phases $i$ and $i+1$ are both good, server $1$ moves between the arrivals of $r_i^*$ and $r_{i+1}^*$, and these moves happen in disjoint time intervals.
Since each phase that is not good belongs to at most two of the pairs $(i,i+1)$, at least $m-1-2\lfloor m/4\rfloor\ge m/4$ of these pairs consist of two good phases when $m\ge4$, and all $m-1\ge m/4$ do when $m\in\{2,3\}$, where every phase is good.
The cost is therefore at least $(m/4)v_1=v_2/4$.

For $\ell>2$, the restriction of $\xi$ to a good phase is a solution for a completed $(\ell-1)$-phase that is not contaminated for it, so by induction it costs at least $v_{\ell-1}/2^{\ell-1}$.
At least $3m/4$ phases are good, so $\xi$ costs at least $\frac{3m}4\cdot\frac{v_{\ell-1}}{2^{\ell-1}}=\frac{3v_\ell}{2^{\ell+1}}\ge\frac{v_\ell}{2^\ell}$.
\end{proof}

\begin{corollary}[after~{\cite[Corollary~20]{BMC26}}]\label[corollary]{cor:opt}
Every solution for a completed $(k,\emptyset)$-phase, from any initial configuration, costs at least $v_k/2^k$.
\end{corollary}

\begin{proof}
No solution is contaminated for a phase of level $k$, since there is no $j>k$; apply \cref{lem:phasecost}.
\end{proof}

\subsection{Proof of \texorpdfstring{\cref{thm:weighted}}{Theorem 1.3}}\label{sec:weighted-proof}

For $k=1$, moving the server whenever necessary is optimal.
If $w_1=0$, moving server~$1$ to $r_t[1]$ at every request serves everything at cost~$0$.
So let $k\ge2$ and $w_1>0$.

\emph{Rounding the weights.}
Let $v_1=w_1$ and $v_i=v_{i-1}\max\{2,\lceil w_i/v_{i-1}\rceil\}$ for $i\ge2$.
Then $v_{i+1}/v_i\in\{2,3,\dots\}$, and $w_i\le v_i\le2^{i-1}w_i$: the first inequality is immediate, and inductively $v_i\le\max\{2v_{i-1},w_i+v_{i-1}\}\le\max\{2^{i-1}w_i,(1+2^{i-2})w_i\}=2^{i-1}w_i$.
Run the algorithm below with the weights $v$; its true cost is at most its cost under $v$, and evaluating an optimal solution under $v$ shows $\OPT_v\le2^{k-1}\OPT_w$.

\emph{Enough points.}
Add to each $M_\ell$ points that are never requested, so that $|M_\ell|\ge d_\ell-1$.
Run the algorithm on the enlarged spaces, and keep each real server at the image of its simulated position under the map that fixes the original points and sends the added points to a fixed original point.
The real servers then serve every request served by the simulated ones and never move more, and the enlargement does not increase $\OPT$.

\emph{The algorithm.}
Run the strategy of \cref{lem:strategy} for $(k,\emptyset)$-phases repeatedly, each call starting with the request at which the previous phase stopped.
By \cref{lem:basic}(a,b), every call accepts a request, and the number $n_c$ of calls on a finite request sequence $\rho$ does not depend on any random choice.
All calls except possibly the last serve completed phases, so $\E[\ALG_v(\rho)]\le n_cc_kv_k$, while \cref{cor:opt}, applied to the restrictions of an optimal solution to the completed phases, gives $\OPT_v(\rho)\ge(n_c-1)v_k/2^k$.
Hence $\E[\ALG_v(\rho)]\le2^kc_k\OPT_v(\rho)+c_kv_k$, and with~\eqref{eq:ck}, $\OPT_v\le2^{k-1}\OPT_w$ and $v_k\le2^{k-1}w_k$,
\[
\E[\ALG_w(\rho)]\le2^{2k-1}\cdot10\cdot5^{k-1}\cdot\OPT_w(\rho)+10\cdot5^{k-1}\cdot2^{k-1}w_k=20^k\OPT_w(\rho)+10^kw_k.
\]

\section{Concluding remarks}

\Cref{thm:upper,thm:lower,thm:weighted} determine the randomized competitive ratio of generalized $k$-server on uniform metrics: $\Theta(k\log k)$ with unit distances and $\exp(\Theta(k))$ with weights.
The lower bound for unit distances uses spaces with $k$ points, while the $\Omega(k)$ bound of~\cite{BJS19} needs only two; it would be interesting to determine the ratio as a function of both $k$ and the number of points per space.
The leading constant with unit distances remains between $\frac1{20}$ and $1$, and the base of the exponent with weights between $2$ and $20$.
For deterministic algorithms, the competitive ratio on uniform metrics remains between $2^k-1$ and $O(k2^k)$~\cite{KT04,BEKN18}, and on general metrics no $f(k)$-competitive algorithm is known for $k\ge3$.
Finally, our algorithms maintain distributions with very large supports; algorithms with polynomial running time would be interesting, as Coester and Cosson~\cite{CC26} recently obtained for randomized $k$-server, with ratio $O(\log n\log^2k)$ on $n$-point metrics.

\section*{AI usage disclosure}
\addcontentsline{toc}{section}{AI usage disclosure}

AI agents were used for exploratory discussions during the development of the proofs and to assist with drafting and revising the manuscript.
The author takes full responsibility for the paper's contents and correctness.

\bibliographystyle{alpha}
\bibliography{refs}

\newcommand{\etalchar}[1]{$^{#1}$}
\begin{thebibliography}{BDBK{\etalchar{+}}94}

\bibitem[AC21]{AC21}
Nikhil Ayyadevara and Ashish Chiplunkar.
\newblock The randomized competitive ratio of weighted $k$-server is at least
  exponential.
\newblock In {\em Proceedings of the 29th Annual European Symposium on
  Algorithms (ESA)}, volume 204 of {\em LIPIcs}, pages 9:1--9:11, 2021.

\bibitem[ACO08]{AC08}
Dimitris Achlioptas and Amin Coja-Oghlan.
\newblock Algorithmic barriers from phase transitions.
\newblock In {\em Proceedings of the 49th IEEE Symposium on Foundations of
  Computer Science (FOCS)}, pages 793--802, 2008.

\bibitem[ACS24]{ACS24}
Nikhil Ayyadevara, Ashish Chiplunkar, and Amatya Sharma.
\newblock A decomposition approach to the weighted $k$-server problem.
\newblock In {\em Proceedings of the 44th IARCS Annual Conference on
  Foundations of Software Technology and Theoretical Computer Science
  (FSTTCS)}, volume 323 of {\em LIPIcs}, pages 6:1--6:17, 2024.

\bibitem[AM06]{AM06}
Dimitris Achlioptas and Cristopher Moore.
\newblock Random $k$-{SAT}: Two moments suffice to cross a sharp threshold.
\newblock {\em SIAM Journal on Computing}, 36(3):740--762, 2006.

\bibitem[BB00]{BB00}
Avrim Blum and Carl Burch.
\newblock On-line learning and the metrical task system problem.
\newblock {\em Machine Learning}, 39(1):35--58, 2000.

\bibitem[BCL{\etalchar{+}}18]{BCLLM18}
S{\'e}bastien Bubeck, Michael~B. Cohen, Yin~Tat Lee, James~R. Lee, and
  Aleksander M{\k{a}}dry.
\newblock $k$-server via multiscale entropic regularization.
\newblock In {\em Proceedings of the 50th Annual ACM Symposium on Theory of
  Computing (STOC)}, pages 3--16, 2018.

\bibitem[BCR23]{BCR23}
S{\'e}bastien Bubeck, Christian Coester, and Yuval Rabani.
\newblock The randomized $k$-server conjecture is false!
\newblock In {\em Proceedings of the 55th Annual ACM Symposium on Theory of
  Computing (STOC)}, pages 581--594, 2023.
\newblock Full version: arXiv:2211.05753.

\bibitem[BDBK{\etalchar{+}}94]{BBKTW94}
Shai Ben-David, Allan Borodin, Richard~M. Karp, G{\'a}bor Tardos, and Avi
  Wigderson.
\newblock On the power of randomization in on-line algorithms.
\newblock {\em Algorithmica}, 11(1):2--14, 1994.

\bibitem[BEK17]{BEK17}
Nikhil Bansal, Marek Eli{\'a}{\v{s}}, and Grigorios Koumoutsos.
\newblock Weighted $k$-server bounds via combinatorial dichotomies.
\newblock In {\em Proceedings of the 58th IEEE Symposium on Foundations of
  Computer Science (FOCS)}, pages 493--504, 2017.

\bibitem[BEKN18]{BEKN18}
Nikhil Bansal, Marek Eli{\'a}{\v{s}}, Grigorios Koumoutsos, and Jesper
  Nederlof.
\newblock Competitive algorithms for generalized $k$-server in uniform metrics.
\newblock In {\em Proceedings of the 29th Annual ACM-SIAM Symposium on Discrete
  Algorithms (SODA)}, pages 992--1001, 2018.
\newblock Journal version: ACM Transactions on Algorithms 19(1):8:1--8:15,
  2023.

\bibitem[BJS19]{BJS19}
Marcin Bienkowski, {\L}ukasz Je{\.z}, and Pawe{\l} Schmidt.
\newblock Slaying {H}ydrae: Improved bounds for generalized $k$-server in
  uniform metrics.
\newblock In {\em Proceedings of the 30th International Symposium on Algorithms
  and Computation (ISAAC)}, volume 149 of {\em LIPIcs}, pages 14:1--14:14,
  2019.

\bibitem[BMC26]{BMC26}
Adithya Bijoy, Ankit Mondal, and Ashish Chiplunkar.
\newblock Weighted $k$-server admits an exponentially competitive algorithm.
\newblock In {\em Proceedings of the 37th Annual ACM-SIAM Symposium on Discrete
  Algorithms (SODA)}, pages 4188--4208, 2026.
\newblock Full version: arXiv:2507.12130v2.

\bibitem[CC26]{CC26}
Christian Coester and Romain Cosson.
\newblock Randomized $k$-server in polynomial time.
\newblock In {\em Proceedings of the 53rd International Colloquium on Automata,
  Languages, and Programming (ICALP)}, volume 374 of {\em LIPIcs}, pages
  65:1--65:20, 2026.

\bibitem[CFK21]{CFK20}
Dimitris Christou, Dimitris Fotakis, and Grigorios Koumoutsos.
\newblock Memoryless algorithms for the generalized $k$-server problem on
  uniform metrics.
\newblock In {\em Proceedings of the 18th Workshop on Approximation and Online
  Algorithms (WAOA 2020)}, volume 12806 of {\em Lecture Notes in Computer
  Science}, pages 143--158, 2021.

\bibitem[CHL{\etalchar{+}}23]{CHLPRSY23}
Marek Chrobak, Samuel Haney, Mehraneh Liaee, Debmalya Panigrahi, Rajmohan
  Rajaraman, Ravi Sundaram, and Neal~E. Young.
\newblock Online paging with heterogeneous cache slots.
\newblock In {\em Proceedings of the 40th International Symposium on
  Theoretical Aspects of Computer Science (STACS)}, volume 254 of {\em LIPIcs},
  pages 23:1--23:24, 2023.

\bibitem[CKZ26]{CKZ26}
Christian Coester, Elias Koutsoupias, and Marek Zbysi{\'n}ski.
\newblock The $k$-server conjecture is true.
\newblock arXiv:2609.15979, 2026.

\bibitem[CL92]{CL91}
Marek Chrobak and Lawrence~L. Larmore.
\newblock The server problem and on-line games.
\newblock In {\em On-Line Algorithms, Proceedings of a DIMACS Workshop},
  volume~7 of {\em DIMACS Series in Discrete Mathematics and Theoretical
  Computer Science}, pages 11--64. American Mathematical Society, 1992.

\bibitem[CV15]{CV15}
Ashish Chiplunkar and Sundar Vishwanathan.
\newblock Metrical service systems with multiple servers.
\newblock {\em Algorithmica}, 71(1):219--231, 2015.

\bibitem[CV20]{CV20}
Ashish Chiplunkar and Sundar Vishwanathan.
\newblock Randomized memoryless algorithms for the weighted and the generalized
  $k$-server problems.
\newblock {\em ACM Transactions on Algorithms}, 16(1):14:1--14:28, 2020.

\bibitem[FR94]{FR94}
Amos Fiat and Moty Ricklin.
\newblock Competitive algorithms for the weighted server problem.
\newblock {\em Theoretical Computer Science}, 130(1):85--99, 1994.

\bibitem[GVW10]{GVW10}
Sascha Geulen, Berthold V{\"o}cking, and Melanie Winkler.
\newblock Regret minimization for online buffering problems using the weighted
  majority algorithm.
\newblock In {\em Proceedings of the 23rd Conference on Learning Theory
  (COLT)}, pages 132--143, 2010.

\bibitem[HW98]{HW98}
Mark Herbster and Manfred~K. Warmuth.
\newblock Tracking the best expert.
\newblock {\em Machine Learning}, 32(2):151--178, 1998.

\bibitem[KP95]{KP95}
Elias Koutsoupias and Christos~H. Papadimitriou.
\newblock On the $k$-server conjecture.
\newblock {\em Journal of the ACM}, 42(5):971--983, 1995.

\bibitem[KT04]{KT04}
Elias Koutsoupias and David~Scot Taylor.
\newblock The {CNN} problem and other $k$-server variants.
\newblock {\em Theoretical Computer Science}, 324(2--3):347--359, 2004.

\bibitem[LW94]{LW94}
Nick Littlestone and Manfred~K. Warmuth.
\newblock The weighted majority algorithm.
\newblock {\em Information and Computation}, 108(2):212--261, 1994.

\bibitem[MMS90]{MMS90}
Mark~S. Manasse, Lyle~A. McGeoch, and Daniel~D. Sleator.
\newblock Competitive algorithms for server problems.
\newblock {\em Journal of Algorithms}, 11(2):208--230, 1990.

\bibitem[New91]{New91}
Lee~A. Newberg.
\newblock The $k$-server problem with distinguishable servers.
\newblock Master's thesis, University of California, Berkeley, 1991.

\bibitem[Sit14]{Sit14}
Ren{\'e} Sitters.
\newblock The generalized work function algorithm is competitive for the
  generalized 2-server problem.
\newblock {\em SIAM Journal on Computing}, 43(1):96--125, 2014.

\bibitem[SS06]{SS06}
Ren{\'e} Sitters and Leen Stougie.
\newblock The generalized two-server problem.
\newblock {\em Journal of the ACM}, 53(3):437--458, 2006.

\end{thebibliography}

\end{document}